\documentclass[
    aps, 
    prr, 
    twocolumn, 
    groupedaddress, 
    reprint, 
    floatfix, 
    nofootinbib, 
    longbibliography
]{revtex4-2}
\pdfoutput=1
\usepackage[utf8]{inputenc}
\usepackage[T1]{fontenc}
\usepackage{amsmath,amssymb,amsthm}
\usepackage{mathtools}
\usepackage{amsthm}
\usepackage{bm}
\usepackage{hyperref}
\usepackage{physics}
\usepackage{float}
\usepackage{xcolor}
\usepackage[shortlabels]{enumitem}
\usepackage[normalem]{ulem}
\usepackage{caption}
\usepackage{subcaption}
\usepackage{orcidlink}

\newtheorem{theorem}{Theorem}
\newtheorem{lemma}{Lemma}
\newtheorem{conjecture}{Conjecture}
\newtheorem{observation}{Observation}
\newtheorem{proposition}{Proposition}
\newtheorem{claim}{Claim}
\newtheorem{corollary}{Corollary}
\theoremstyle{definition}
\newtheorem{definition}{Definition}
\theoremstyle{remark}
\newtheorem{remark}{Remark}
\usepackage{graphicx}

\newcommand{\one}{\mathbf{1}}
\newcommand{\RR}{\mathbb{R}}
\newcommand{\NN}{\mathbb{N}}

\begin{document}

\title{Star Topology Optimizes the Charging Power of Quantum Batteries}

\author{Matthieu Sarkis\orcidlink{https://orcid.org/0009-0002-5494-8406}}
\email[]{matthieu.sarkis@uni.lu}
\affiliation{Department of Physics and Materials Science, University of Luxembourg, L-1511 Luxembourg City, Luxembourg}

\author{Oskar A. Pro\'sniak\orcidlink{https://orcid.org/0000-0003-4550-7561}}
\email[]{oskar.prosniak@uni.lu}
\affiliation{Department of Physics and Materials Science, University of Luxembourg, L-1511 Luxembourg City, Luxembourg}

\author{Samuel Nigro}
\email[]{samuel.nigro@ens-paris-saclay.fr}
\affiliation{Université Paris-Saclay, ENS Paris-Saclay, DER de Physique, 91190, Gif-sur-Yvette, France}

\author{Alexandre Tkatchenko\orcidlink{https://orcid.org/0000-0002-1012-4854}}
\email[]{alexandre.tkatchenko@uni.lu}
\affiliation{Department of Physics and Materials Science, University of Luxembourg, L-1511 Luxembourg City, Luxembourg}

\begin{abstract}
    Quantum batteries are quantum systems that store energy and deliver it on demand, and their practical value hinges on how fast they can be charged. While collective charging protocols and global control are known to enhance charging power, it remains unclear how the battery’s internal interaction architecture itself constrains performance. Here we study interacting fermionic batteries whose internal couplings are encoded by a graph adjacency matrix, charged via a simple interaction with an external fermionic device. We prove that the \emph{star} topology maximises the early time charging power, which proxies the maximal average power -- a widely used quantum battery quality metric. We substantiate the result numerically by an exhaustive sweep over all graphs with $N\leq 7$ vertices and by benchmarks against random graph ensembles at larger $N$. Our findings shed light on architecture as a controllable knob for fast charging and motivate hub-and-spoke designs in scalable quantum-battery platforms.
\end{abstract}

\maketitle

\section{Introduction}
    Quantum batteries (QBs) are defined as quantum systems engineered to store energy and deliver it on demand, leveraging the laws of quantum mechanics. Since the seminal observation that global (entangling) unitary control can enhance the extractable work and the rate of energy deposition compared to independent charging of subsystems~\cite{alicki2013entanglement,ferraro2018highpower, andolina2019extractable}, the field has developed into a vibrant subarea of quantum thermodynamics that connects quantum control, many-body physics, and open-system engineering; see the comprehensive recent overview in Ref.~\cite{campaioli2024colloquium} and related perspectives~\cite{friis2018gaussian,barra2019dissipative}. A core question is whether, and under what constraints, collective effects can produce genuine scaling advantages in the charging power--the work per unit time--over strategies that address the battery cells in parallel or sequentially.

    Early rigorous bounds clarified the resource requirements behind such advantages. In particular, Ref.~\cite{campaioli2017enhancing} derived limits on the achievable ``quantum advantage'' in charging power for constrained driving, linking performance improvements to the \emph{order} of interactions that the charger can induce across the $N$ cells~\cite{juliafarre2020bounds, gyhm2021advantage}, and showing that extensive scaling cannot be achieved without global operations~\cite{gyhm2022cannot}. A complementary geometric viewpoint yields minimal-time charging bounds that further tighten constraints on mean power~\cite{gyhm2024minimal}. Subsequent models demonstrated that superextensive power scaling can arise either from coherent shortcuts in Hilbert space or from effective mean-field mechanisms in many-body platforms, including spin chains and long-range interacting systems~\cite{le2018spin,andolina2019versus,rossini2019mbl,rossini2020syk}. At the same time, microscopic open-system analyses have shown that collective enhancements can persist---and can even be assisted---in realistic dissipative settings~\cite{carrasco2022collective,barra2019dissipative,farina2019charger,gherardini2020stabilizing,qi2021magnon,pokhrel2025dissipative}, while collision-model approaches established that quantum coherence of the chargers yields a genuine single-cell speed-up over all incoherent strategies~\cite{seah2021quantum}, and indefinite causal order has been shown to further enhance charging strategies~\cite{zhu2023ico}. Beyond bounds, explicit many-body models demonstrate that superextensive power and enhanced stability can emerge from delocalized modes, long-range couplings, or chaotic dynamics. In particular, fast-scrambling fermionic models of Sachdev--Ye--Kitaev (SYK) type were shown to support a quantum charging advantage and to exhibit exceptional stability of the stored energy after charging~\cite{rossini2020sykadvantage, rosa2020ultrastable, divi2025randomwalkgraphs}. Other fermionic platforms---emphasizing particle statistics and band-structure effects---have also been explored~\cite{konar2022ultracold,grazi2025chargingfreefermion,grazi2025universal}. Together, these results suggest that both the interaction topology and the nonequilibrium resources available to the charger are decisive ingredients for maximizing charging power.

    The previous results found in the literature often focused on the interaction structure introduced via the charging protocol, disregarding the internal structure of the battery itself. Important exceptions are Refs.~\cite{konar2022ultracold,grazi2025chargingfreefermion,grazi2025universal}, which, however, investigate only the regular lattice structure of the internal interactions, and general results completely agnostic to the battery architecture~\cite{gyhm2022cannot,caravelli_energy_2021,zakavati_bounds_2021}. In this work we analyze the influence of the internal topology on the battery quality. In particular, we consider the battery consisting of $N$ fermionic degrees of freedom interacting with uniform strength according to a graph structure $G$. The charging is modelled by a quenched interaction with an external fermionic bath. We pick the maximal average power to be our main quality metric of interest. Our analysis shows that the \emph{star} topology is both optimal and stable against perturbation.

    The star architecture has additional practical appeal. It is compatible with widely used physical platforms where a central mode or mediator couples identically (or near-identically) to an ensemble of cells: examples include cavity and circuit QED (a single electromagnetic mode addressing many qubits), trapped-ion chains driven by a common phonon, and magnonic or mechanical resonators coupled to spin or superconducting degrees of freedom.
    Such hub-and-spoke layouts have already featured in theoretical and experimental proposals for QBs with harmonic or micromaser chargers, collisional schemes, and dissipative cavity setups~\cite{seah2021quantum,carrasco2022collective,shaghaghi2022micromasers,shaghaghi2023lossy,rodriguez2023aidiscovery,rojofrancas2024stablecollective}.
    Our results provide a possible explanation for their good performance. 

    We emphasize that, while the particular choice of the system was dictated by its analytical tractability and experimental relevance, we expect our findings to generalize to charging protocols and microscopic models beyond those considered here.

    The remainder of the paper is organized as follows. In Sec.~\ref{sec:model} we introduce the model, fix the graph-theoretic conventions used to encode the battery architecture, and state the main results. Sec.~\ref{sec:math} provides the analytical backbone of the work: we derive an early-time expansion of the charging dynamics, reduce the relevant power functional to spectral quantities of the underlying graph, and prove the optimality of the star topology first in the absence of a local field and then for a general nonzero local field. In Sec.~\ref{sec:num} we complement these proofs with numerical evidence, combining broad searches over graph families with direct simulations of the microscopic model to benchmark performance and assess the absence of close competitors. Finally, Sec.~\ref{sec:conclusion} concludes and summarizes the implications for topology-driven quantum-battery design and outlines directions for extending the approach to other architectures, charging scenarios, and platforms.

\section{Model and results}\label{sec:model}

    \subsection{Graphs and notations}
        Unless stated otherwise, graphs $G=(V,E)$ are simple, undirected, and unweighted. The adjacency $A\in\{0,1\}^{N\times N}$ is symmetric with zero diagonal; $|V|=N$ and $|E|=M$. We write $\mathbf{1}=(1,\dots,1)^\top$, and use $\{u_k\}_{k=1}^N$ for an orthonormal eigenbasis of $A$ in the Euclidean inner product $\langle x,y\rangle=x^\top y$. The average degree is $\bar d:=2M/N$. By the (adjacency) Rayleigh quotient of a unit vector $u$ we mean $u^\top A u$.

    \subsection{Microscopic model}
        We will consider a system composed of a collection of $N$ fermionic particles, called hereafter the \textit{battery} and described by a Fock space $\mathcal F_\textsc{b} =  \bigoplus_{i=0}^N \bigwedge^i \mathbb C^N \simeq \left( \mathbb{C}^2\right)^{\otimes N}$. We adjoin to the battery a charging device composed of $L$ fermions, hereafter referred to as the \textit{charger}, described by a Fock space $\mathcal F_\textsc{c} \simeq \left(\mathbb{C}^2\right)^{\otimes L}$. The total system is then described by a tensor product of the two Fock spaces $\mathcal F=\mathcal F_\textsc{c}\otimes \mathcal F_\textsc{b}$. We endow the total system with the following unitary dynamics:
        \begin{equation}
            H = H_\textsc{b} + H_\textsc{c} + H_\text{int},
        \end{equation}
        with the following quadratic battery Hamiltonian:
        \begin{equation}
            H_\textsc{b} = \frac{H_\textsc{b}^0 - \bar E}{\Delta/2},
        \end{equation}
        with ($h\geq 0$)
        \begin{equation}
            H_\textsc{b}^0 = h\sum_{i\in V}c_i^\dagger c_i + \frac{N}{M}\sum_{\{i,j\}\in E}A_{ij}\left(c_i^\dagger c_j + \text{h.c.}\right),
        \end{equation}
        and $\bar E = \frac{E_\text{max} + E_\text{min}}{2}$, $\Delta = E_\text{max} - E_\text{min}$ are the average spectral energy and the width of the spectrum, respectively ($E_\text{max}$ and $E_\text{min}$ are the maximal and minimal eigenvalues of $H_\textsc{b}^0$). The fermionic creation and annihilation operators satisfy the usual canonical anti-commutation relations:
        \begin{equation}
            \{c_i, c_j^\dagger\} = \delta_{ij}\,.
        \end{equation}

        The couplings $A_{ij}$ will be taken to be valued in $\{0, 1\}$, endowing $A$ with an interpretation of an adjacency matrix of a graph $G(V,E)$ with $|V|=N$ vertices and $|E|=M$ edges.
        
        The reason for this restriction in the range of the coupling constants is to focus solely on the role of connectivity.

        The charger Hamiltonian is chosen to be a fermionic bath with positive energies
        \begin{equation}
            H_\textsc{c} = \sum_{i=1}^L \omega_i d_i^\dagger d_i,\quad \omega_i>0,
        \end{equation}
        with
        \begin{equation}
            \{c_i,d_j\}=\{c_i,d_j^\dagger\}=0\,.
        \end{equation}
        The interaction (or \textit{charging}) term is chosen to be uniform:
        \begin{equation}
            H_\text{int} = \kappa\sum_{i=1}^{N} \sum_{j=1}^L \left(c_i^\dagger d_j + d_j^\dagger c_i\right).
        \end{equation}
        The system is initialized in the product state $|\psi(0)\rangle = |0\rangle\otimes|\varnothing\rangle$ where $|0\rangle$ is the Fock vacuum of the charger and $|\varnothing\rangle$ is the ground state of the battery Hamiltonian. Observe that $|\psi(0)\rangle$ is in particular the ground state of the non-interacting model. The system then evolves according to the full Hamiltonian, and we keep track of the state of the battery $\rho_\textsc{b}(t) = \Tr_C\ketbra{\psi(t)}{\psi(t)}$, and, by corollary, of various quality metrics.

        The quality of a quantum battery can be measured through various energetic observables. Among the standard observables, the stored energy, the ergotropy, the maximal charging power and the energy fluctuations are often considered in the literature. In this work, we will focus on the maximal average charging power, which is the primary quantity of interest in the context of quantum advantage in quantum thermodynamics. The average charging power is defined as:
        \begin{equation}
        \label{eq:instant_power}
            P(t) = \frac{1}{t}\left(\text{Tr}\big[H_\textsc{b}\rho_\textsc{b}(t)\big] - \Tr[H_\textsc{b}\rho_\textsc{b}(0)]\right).
        \end{equation}
        The maximal average charging power is defined as its supremum over all times
        \begin{equation}
            P_\text{max} = \sup_{t\geq 0}P(t).
        \end{equation}
        Equipped with the above definitions, we are now ready to state the main results of this work.

    \subsection{Results}

    The main result of this work is a partial proof of the following claim.

    \begin{claim}
        Among all quantum batteries described in the previous subsection, together with their charging protocols, for fixed values of parameters $(N,L,h,\lbrace\omega_i\rbrace,\kappa)$, both
        \begin{enumerate}[(i)]
            \item the early time average charging power $P(0^+)$,
            \item the maximal average charging power $P_\text{max}$,
        \end{enumerate}
         are maximised by the quantum battery with a star graph $S_N$ architecture.
    \end{claim}

    For $h=0$ we rigorously prove part $(i)$ in the above claim (Theorem~\ref{thm:star-h0}). We numerically show that there is a high positive correlation between the values of $P(0^+)$ and $P_\text{max}$, and confirm the claim by both a full numerical sweep over graphs with $N=7$ vertices and tests of random families of graphs.

\section{Mathematical analysis}
\label{sec:math}

    We derive an early–time expression for the charging power and reduce its maximization to a
    spectral problem on the interaction graph. Diagonalizing the adjacency $A$ yields normal
    modes with eigenvalues $\{\varepsilon_k\}$ and uniform overlaps $v_k:=\mathbf{1}^\top u_k$,
    which control the relevant bounds. In the zero–field case $h=0$ this reduction becomes
    particularly transparent and leads to a sharp optimality statement for the star graph.

    \subsection{Early-time expansion and spectral reduction}

        We focus on the average power given in eq. (\ref{eq:instant_power}) and consider the early-time regime. Up to $O(t^2)$ corrections
        \begin{equation}
        \label{eq:P-early}
            P(t)=-\frac{t}{2}\,\big\langle\,[H_\textrm{int},[H_\textrm{int},H_\textsc{b}]]\,\big\rangle_{\psi(0)}.
        \end{equation}

        \paragraph*{Diagonalization of $A$ and mode weights.}
        Let $A=O\,\mathrm{diag}(\varepsilon_1,\dots,\varepsilon_N)\,O^\top$ with $O\in O(N)$ and orthonormal eigenvectors $u_k$ (columns of $O$), $Au_k=\varepsilon_k u_k$. Define the uniform overlaps and weights
        \begin{equation}
            v_k=\mathbf{1}^\top u_k,
            \qquad
            w_k=v_k^2\;\ge 0.
        \end{equation}
        We note here that in the literature $v_k$ are known as fundamental weights of the graph \cite{van_mieghem_graph_2016}. Introduce rotated annihilation operators
        \begin{equation}
            \tilde c_k=\sum_{i=1}^N c_i\,O_{ik},
        \end{equation}
        which form another representation of CAR algebra. One readily checks
        \begin{equation}
        \label{eq:diag-Hb-Hint}
        \begin{aligned}
            H_\textsc{b}^0&=\sum_{k=1}^N\Big(\frac{N}{M}\,\varepsilon_k+h\Big)\,\tilde c_k^\dagger\tilde c_k, \\
            H_{\mathrm{int}}&=\kappa\sum_{k=1}^N\sum_{i=1}^L v_k\big(\tilde c_k^\dagger d_i+d_i^\dagger \tilde c_k\big).
        \end{aligned}
        \end{equation}

        \paragraph*{Active (negative) sector.}
        Set $E_k(h):=\frac{N}{M}\,\varepsilon_k+h$ and define the $h$–dependent set of \emph{active} modes
        \begin{equation}
            K_-(h)=\big\{\,k:\,E_k(h)<0\,\big\},
        \end{equation}
        and let $K_+(h)$ denote its complement. Let $E_{\max}=\sum_{k\in K_+(h)}E_k(h)$ and $E_{\min}=\sum_{k\in K_-(h)}E_k(h)$ be the extremal many–body energies (the $\varepsilon_k$ come in opposite pairs). Inserting~\eqref{eq:diag-Hb-Hint} in~\eqref{eq:P-early} and evaluating on the initial state $|0\rangle\!\otimes\!|\varnothing\rangle$ gives the compact form
        \begin{equation}
        \label{eq:P-early-final}
            P(t)=\frac{t L \kappa^2}{E_{\min}}\,\sum_{k\in K_-(h)} w_k\,E_k(h).
        \end{equation}
        This is the key spectral reduction: only modes with $E_k(h)<0$ contribute, and their weight is precisely the uniform overlap $w_k$. For the future convenience let us isolate the graph contribution and define
        \begin{equation}
            R_G(h) = \frac{1}{E_{\min}}\,\sum_{k\in K_-(h)} w_k\,E_k(h),
        \end{equation}
        so that $P(t) = t L \kappa^2 R_G(h)$ and the task of maximising the early-time power reduces to maximisation of $R_G(h)$.

    \subsection{No local field: $h=0$}
    \label{sec:h0}
        To simplify the notation let us set $K_\pm \equiv K_\pm(0)$ and $R_G\equiv R_G(0)$. We then have
        \begin{equation}
        \label{eq:P-h0}
            R_G=\frac{1}{\varepsilon_+}\!\!\sum_{k\in K_-} w_k\,|\varepsilon_k|,
            \qquad
            \varepsilon_+=\sum_{\ell\in K_+}\varepsilon_\ell.
        \end{equation}

        The simplicity of the expression allows for the rigorous treatment. In the following we show that $R_G$ is maximised by the star graph.

        \begin{lemma}
        \label{lem:neg-sum-bound}
            Let $d_i:=\sum_j A_{ij}$ be the degrees, $d:=(d_1,\dots,d_N)$, and $\|d\|_2$ its Euclidean norm. Then
            \begin{equation}
            \label{eq:neg-sum-bound}
                \frac{\sum_{k\in K_-} w_k\,|\varepsilon_k|}{\varepsilon_+}
                \;\le\;
                \frac{\|d\|_2\,\sqrt{N}}{2\sqrt{M}}\;-\;\sqrt{M}.
            \end{equation}
        \end{lemma}

        \begin{proof}
            We use the identities
            \begin{align*}
                &\text{(i)}\quad \sum_{k\in K_+}w_k\,\varepsilon_k-\sum_{k\in K_-}w_k\,|\varepsilon_k|
                =\sum_k w_k\,\varepsilon_k\\
                &\hspace{10em}=\mathbf{1}^\top A\mathbf{1}
                =2M,\\
                &\text{(ii)}\quad \sum_k w_k\,\varepsilon_k^2
                =\mathbf{1}^\top A^2\mathbf{1}
                =\sum_i d_i^2
                =\|d\|_2^2,\\
                &\text{(iii)}\quad \sum_k w_k
                =\|\mathbf{1}\|_2^2
                =N,\\
                &\text{(iv)}\quad \varepsilon_+\;\ge\;\sqrt{M},
            \end{align*}
            where the last item comes from \cite{jiongsheng_lower_1998}. From (i), we have
            \begin{equation}
            \begin{aligned}
                \sum_{k\in K_-}w_k\,|\varepsilon_k|
                &=\tfrac12\Big(\sum_k w_k\,|\varepsilon_k|\;-\;2M\Big) \\
                &\le \tfrac12\Big(\sqrt{\textstyle\sum_k w_k\,\varepsilon_k^2}\,\sqrt{\textstyle\sum_\ell w_\ell}\;-\;2M\Big) \\
                &\le \tfrac12\big(\|d\|_2\,\sqrt{N}-2M\big),
            \end{aligned}
            \end{equation}
            where we used Cauchy–Schwarz and (ii)–(iii). Dividing by (iv) gives~\eqref{eq:neg-sum-bound}.
        \end{proof}

        \begin{theorem}[Star optimality at $h=0$]
        \label{thm:star-h0}
            For fixed $N\geq 2$, the right-hand side of~\eqref{eq:neg-sum-bound} is maximized (and the bound is saturated) by the star graph. Consequently, at early times and $h=0$ the star attains the largest $R_G$ among graphs with the same number of vertices.
        \end{theorem}

        \begin{proof}
            A sharp bound on the degree norm (see, e.g., standard spectral graph inequalities~\cite{de_caen_upper_1998}) states that for fixed $N,M$,
            \begin{equation}
                \|d\|_2^2\;\le\;M\Big(\frac{2M}{N-1}+N-2\Big).
            \end{equation}
            Hence
            \begin{equation}
                \frac{\|d\|_2\,\sqrt{N}}{2\sqrt{M}}\;-\;\sqrt{M}
                \;\le\;
                \sqrt{N}\,\frac{\sqrt{\tfrac{2M}{N-1}+\,N-2}}{2}\;-\;\sqrt{M},
            \end{equation}
            and the right-hand side is monotonically decreasing in $M$. The connectivity constraint implies $M\ge N-1$, so the maximum is achieved at $M=N-1$. A direct check shows that the star saturates the bound at $M=N-1$. Finally, among disconnected graphs with $\sum_r M_r\le N-1$ edges and $\sum_r N_r=N$ vertices, the analogous component-wise bound is strictly worse than the connected-star value (algebra omitted), so no disconnected graph can outperform the star.
        \end{proof}

        \begin{remark}
            At $h=0$, early-time power rewards two features simultaneously: (i) large \emph{magnitudes} of negative eigenvalues (to stay in the active sector), and (ii) large \emph{uniform overlaps} $w_k$ of the corresponding eigenvectors. The hub-and-spoke (star) concentrates connectivity into one mediator, producing (a) an extremal spectral radius and (b) an eigenvector structure that carries a large uniform component on the negative side. Both ingredients enter~\eqref{eq:P-h0}, explaining the optimality.
        \end{remark}

    \subsection{Generalization to a non-trivial local field: $h>0$}
    \label{subsec:h_non_trivial}

        In order to be able to physically interpret the system as being composed of $N$ local cells, each endowed with its own intrinsic dynamics, it is important to reinstate a non-zero local energy stored in individual cells. The finite-field regime $h>0$ is a smooth deformation of the zero-field analysis: we keep the eigenbasis of the interaction graph fixed, so all conventions remain unchanged. Turning on $h$ shifts mode frequencies and reweights the eigenvalue-weighted sums that control our power bounds.

        \begin{observation}
            The quantity $R_G(h)$ may be viewed as a weighted average in the following way
            \begin{equation}
                    \label{eq:convex-comb}
            \begin{aligned}
                R_G(h)&=\sum_{k\in K_-(h)} p_k(h)\,w_k, \\
                p_k(h)&:=\frac{E_k(h)}{\sum_{j\in K_-(h)}E_j(h)}\in(0,1],\\ \sum_{k\in K_-}p_k(h)&=1.
            \end{aligned}
            \end{equation}
        \end{observation}

        Our special attention deserves the following quantity $w_\mathrm{min} := \sum_{k'} w_{k'}$, where $k'$ are such that $E_{k'}(h) = \min_{k} E_{k'}(h)$, and which is invariant under the unitary symmetry of the eigenspace. We devote our special attention to it because of two reasons: first, because it receives the largest weight $p_{k'}(h)$, and second, because $k'$ belong to $K_{-}(h)$ for the largest window of values of $h$. While we cannot prove that $R_G(h)$ is governed by $w_\mathrm{min}$ we are going to provide arguments consistent with such a statement and moreover pointing towards overall optimality of a star graph.

        \begin{observation}[Star architecture]\label{prop:star}
            For the star $S$ with $M=N-1$ and spectrum $\{\sqrt{N-1},-\sqrt{N-1},0,\dots,0\}$, the negative sector consists of the single eigenvalue $-\sqrt{N-1}$ for $0\le h<h_\star:=\frac{N}{\sqrt{N-1}}$. Hence
            \begin{equation}
            \label{eq:star-flat}
                R_{S}(h)\equiv w_\textrm{min}{(S)}=\frac{N}{2} - \sqrt{N-1}.
            \end{equation}
        \end{observation}

        We would like to compare now $w_\textrm{min}$ for different graphs and show the optimality of the star. After performing exhaustive numerical check of graphs with $N\leq7$ and sample check of larger random graphs we arrived with the following claim, which however we were unable to prove. This is our primary evidence of optimality of the star graph also in the presence of a non-trivial local field $h>0$.

        \begin{conjecture}[Star optimality for the uniform overlap]\label{conjecture}
        Let $G$ be a non-empty graph on $N$ vertices. Its minimal energy uniform overlap $w_\mathrm{min}(G)$ is upper bounded by
            \begin{equation}
                w_\mathrm{min}(G) \leq w_\mathrm{min}(S) = \frac{N}{2} - \sqrt{N-1}.
            \end{equation}
        \end{conjecture}

        In the remainder of this subsection we present a case study of certain families of graphs and rigorously show that for all graphs belonging to these families $R_G(h)$ is upper bounded by $R_{S}(h)$. Our analysis is based on the following observation.

        \begin{observation}\label{obs: envelope bound}
            From the fact that $R_G(h)$ is a weighted mean it follows that for any graph $G$ and any $h>0$
            \begin{equation}
                R_G(h) \leq w_-(G) \equiv \max_{k\in K_-}( w_k).
            \end{equation}
            Moreover, the inequality is saturated e.g.\ by a star graph.
        \end{observation}
        
        \subsubsection{Expanders and (almost-)regular expanders are not serious competitors}
        \label{sec:expanders-no-star}

            One may possibly worry in general about expanders. Expanders are sparse yet highly connected graphs: random walks mix rapidly and small vertex sets have many outgoing edges. Heuristically, this means eigenvectors tend to be \emph{delocalized}. Since $R_G(h)$ rewards large $w_k$ in the negative sector, one might suspect that expander families, by being both delocalized and capable of having sizeable negative eigenvalues, could challenge or even exceed the star benchmark. We show below that this does \emph{not} occur. For regular expanders the relevant overlap is \emph{exactly zero}; for almost-regular expanders we obtain a sharp, explicit bound in terms of degree variance that remains far below the star for any reasonable notion of ``near-regularity.''

            \begin{definition}[Regular graphs and spectral expanders]
                A graph $G$ on $N$ vertices is \emph{$d$-regular} if every vertex has degree $d$. Its adjacency spectrum satisfies $d=\epsilon_1\ge \epsilon_2\ge \dots\ge \epsilon_N\ge -d$. A $d$-regular graph is a \emph{spectral expander} if $\max\{|\epsilon_2|,|\epsilon_N|\}\le \epsilon_\star$ with $\epsilon_\star\ll d$; Ramanujan graphs achieve $\epsilon_\star\le 2\sqrt{d-1}$.
            \end{definition}

            \begin{proposition}[Regular graphs have zero uniform overlap in nontrivial modes]\label{prop:reg-zero}
                Let $G$ be connected and $d$-regular with adjacency $A$. Then $A\one=d\,\one$, and for any eigenvector $u_k$ with eigenvalue $\epsilon_k\ne d$ one has $\one^\top u_k=0$. Consequently,
                \begin{equation}
                w_-(G)=0.
                \end{equation}
            \end{proposition}

            \begin{proof}
                Follows from a direct computation.
            \end{proof}

            \begin{remark}
                For $d$-regular \emph{bipartite} graphs one has the eigenvalue $-d$ with an eigenvector that is constant $+1$ on one side and $-1$ on the other; because the sides have equal cardinality in the $d$-regular case, this eigenvector also has zero sum and is orthogonal to $\one$. We will soon study in greater details the case of bipartite graphs.
            \end{remark}

            The preceding proposition already rules out \emph{all} $d$-regular expanders.

            We now quantify how far \emph{almost-regular} expanders remain from the star. One could indeed imagine that \emph{impurities} in the connectivity pattern of the battery could lead to an enhancement of the charging power property of the system. The following result rules out this possibility.

            \begin{lemma}[Degree-variance control of negative-mode uniform overlap]\label{thm:deg-variance}
                Let $d_i$ denote the degree of vertex $i$ and $\bar d:=\frac{1}{N}\sum_i d_i$ the average degree. Define the (relative) root mean squared degree fluctuation
                \begin{equation}
                    \delta_{\mathrm{rms}}=\frac{\Big(\frac{1}{N}\sum_{i=1}^N (d_i-\bar d)^2\Big)^{\!1/2}}{\bar d}.
                \end{equation}
                For any simple graph $G$ with adjacency $A$ and average degree $\bar d>0$,
                \begin{equation}
                    w_-(G) \;\le\;\frac{\sum_{i=1}^N (d_i-\bar d)^2}{\bar d^2}=N\,\delta_{\mathrm{rms}}^2.
                \end{equation}
            \end{lemma}

            \begin{proof}
                Write an orthonormal eigenbasis $A u_k=\epsilon_k u_k$ and expand $\one=\sum_k \alpha_k u_k$, with $\alpha_k=\one^\top u_k$. Since $A\one=d$, decompose $d=\bar d\,\one + (d-\bar d\,\one)$ and take the $u_k$-component:
                \begin{equation}
                    \epsilon_k \alpha_k = u_k^\top d = \bar d\,\alpha_k + u_k^\top(d-\bar d\,\one),
                \end{equation}
                hence
                \begin{equation}
                    \alpha_k = \frac{u_k^\top(d-\bar d\,\one)}{\epsilon_k-\bar d}.
                \end{equation}
                If $\epsilon_k<0$, then $|\epsilon_k-\bar d|\ge \bar d$, so by Cauchy--Schwarz
                \begin{equation}
                    |\alpha_k|\;\le\;\frac{\|\,d-\bar d\,\one\,\|_2}{\bar d}
                    =\frac{\Big(\sum_i(d_i-\bar d)^2\Big)^{1/2}}{\bar d}.
                \end{equation}
                Squaring and maximizing over $K_-$ yields the claimed bound.
            \end{proof}

            \begin{remark}
                 We call $G$ \emph{almost-regular} if $\delta_{\mathrm{rms}}\ll 1$.
            \end{remark}

            \begin{corollary}[No almost-regular expander beats the star]
            \label{cor:almost-regular}
                If $\delta_{\mathrm{rms}}\le \delta_0<1/\sqrt{2}$ uniformly along a family $\{G_N\}$ (in particular for any fixed-degree almost-regular expander family), then for all sufficiently large $N$
                \begin{equation}
                    \frac{w_-(G)}{w_-(S)}\;\le\;2\,\delta_0^2\,(1+o(1))\;<\;1.
                \end{equation}
                If the family is exactly $d$-regular, then $\delta_{\mathrm{rms}}=0$ and $w_-(G)=0$ for all $N$ by Proposition~\ref{prop:reg-zero}.
            \end{corollary}

            \begin{remark}
                Expanders are homogeneous: the ``flat'' direction $\one$ is confined to the \emph{top positive} eigenmode, while all nontrivial modes---including negative ones---are orthogonal to~$\one$. Almost-regularity preserves this picture up to a quantitative factor controlled by degree variance. Therefore, despite their excellent global connectivity and delocalization, \mbox{(almost-)regular} expanders cannot approach, let alone surpass, the star benchmark in the early-time behaviour.
            \end{remark}

        \subsubsection{Complete-bipartite family $K_{p,q}$}
        \label{sec:Kpq}

        As mentioned, we require large uniform overlap $w_k$ in the negative mode sector. There is a natural connection between the size of $w_k$ and negativity of associated eigenvalue. It can be understood more concretely by looking at the sign pattern of eigenvectors.  Though it can appear as a trivial step, splitting the quadratic form into within-part and cross-part contributions reveals how negativity arises from edges bridging opposite signs. This structural view highlights that nearly bipartite graphs are natural candidates for producing negative eigenvectors with nontrivial uniform overlap.

            \begin{lemma}[Sign-split identity]
            \label{lem:signsplit}
                For any $u\in\RR^N$, writing $S^+:=\{i: u_i\ge 0\}$ and $S^-:=\{i: u_i<0\}$, one has
                \begin{equation}
                \label{eq:signsplit}
                \begin{aligned}
                    u^\top A u &= 2\sum_{\{i,j\}\in E(S^+)} u_i u_j\ +2\sum_{\{i,j\}\in E(S^-)} u_i u_j\\
                    &-2\sum_{\{i,j\}\in E(S^+,S^-)} \abs{u_i}\abs{u_j},
                \end{aligned}
                \end{equation}
                where $E(S^+)$ and $E(S^-)$ denote edges internal to $S^+$ and $S^-$, and $E(S^+,S^-)$ the set of cross edges.
            \end{lemma}

            \begin{proof}
                Because $A$ is symmetric with zero diagonal, $u^\top A u = 2\sum_{\{i,j\}\in E} u_i u_j$. Split the edge set into the three disjoint families and note that on cross edges the product $u_i u_j$ has negative sign, so $u_i u_j = -\abs{u_i}\abs{u_j}$.
            \end{proof}

            \begin{remark}
                Within-side edges \emph{help} positivity; cross edges \emph{help} negativity. Thus a negative eigenvector with large uniform sum prefers graphs that look nearly bipartite \emph{relative to its sign pattern}. This nudges us toward complete bipartites as the model family to analyze exactly.
            \end{remark}

            Motivated by the sign-based picture, we now analyze the cleanest bipartite examples: complete bipartite graphs. These allow exact computation of spectra and overlaps, providing a controlled laboratory for testing the intuition. Within this family, we find that the uniform overlap is maximized at the most unbalanced partition—namely, the star. We compute the negative eigenpair and uniform overlap for $K_{p,q}$ and identify the star as the maximizer within this family.

            \begin{definition}[Complete bipartite graph]
                Let $K_{p,q}$ denote the bipartite graph on disjoint parts $U,V$ with $|U|=p$, $|V|=q$, and all $pq$ cross edges present, no within-part edges.
            \end{definition}

            \begin{proposition}[Spectrum and negative eigenvector of $K_{p,q}$]
            \label{prop:Kpq-spectrum}
                The adjacency matrix of $K_{p,q}$ has eigenvalues $\pm\sqrt{pq}$ (simple) and $0$ (with multiplicity $N-2$). A unit negative eigenvector is
                \begin{equation}
                \label{eq:Kpq-neg}
                    u_-=\frac{1}{\sqrt2}\left(\frac{\one_U}{\sqrt p}-\frac{\one_V}{\sqrt q}\right),\qquad \epsilon_-=-\sqrt{pq}.
                \end{equation}
                Moreover
                \begin{equation}
                \label{eq:Kpq-overlap}
                    \big(\one^\top u_-\big)^2=\frac{(\sqrt p-\sqrt q)^2}{2}.
                \end{equation}
            \end{proposition}

            \begin{proof}
                The proof is elementary.
            \end{proof}

            \begin{remark}
                In the complete bipartite world, the optimal negative eigenvector is \emph{exactly} constant (with opposite signs) on the two parts. Its uniform overlap depends only on $p$ and $q$, and grows as the parts become more unbalanced. The most unbalanced partition, $p=1$ or $q=1$, is the star; it maximizes the overlap within this family.
            \end{remark}

            \begin{corollary}[The star as a maximizer]
            \label{cor:star-max}
                For fixed $N$, the map $p\mapsto (\sqrt p-\sqrt{N-p})^2$ is maximized at $p=1$ or $p=N-1$. Therefore among all complete bipartites on $N$ vertices, the star $K_{1,N-1}$ uniquely maximizes the negative-mode uniform overlap.
            \end{corollary}

            \begin{proof}
                Consider $g(p):=(\sqrt p-\sqrt{N-p})^2$ for $p\in[1,N-1]\cap\NN$. Extending to real $p\in(0,N)$, $g'(p)=\frac{(\sqrt{N-p}-\sqrt p)^2}{2\sqrt p\,\sqrt{N-p}(\sqrt p+\sqrt{N-p})}\ge 0$ for $p\le N/2$, so $g$ increases on $[1,\lfloor N/2\rfloor]$ and is symmetric by $p\leftrightarrow N-p$. Hence maxima at $p=1$ or $N-1$.
            \end{proof}

            \begin{remark}
                This is the formal version of a simple picture: with one hub and $N-1$ leaves, the negative eigenvector is as ``polarized'' as possible while still having a large positive mass on many vertices. That polarization furnishes a large negative Rayleigh and simultaneously a large uniform sum.
            \end{remark}

\section{Numerical analysis}
\label{sec:num}

        \subsection{Evidence for the absence of strong challengers}

        The analytical results of Sec.~\ref{subsec:h_non_trivial} reduce the maximization over $h$ to a purely spectral
        envelope: on any $h$-window with fixed active set, the early-time power ratio is a convex
        combination of the weights $w_k$ and satisfies
        $\sup_h R_G(h)\le w_-(G) \equiv \max_{\varepsilon_k<0} w_k$, with an equality for a star graph  (Obs.~\ref{obs: envelope bound}). Thus, any architecture that ever outperforms the star must realize a \emph{negative} eigenvector with even larger uniform overlap than the star’s benchmark $w_{-}(S)$ (Eq.~\ref{eq:star-flat}).

        We probe here a broad set of sparse and heterogeneous topologies across sizes
        $N\in\{10,15,20,30,40,50,60\}$. For each $N$ and model we generate
        $500$ independent samples:
        (i) Erd\H{o}s--R\'enyi $G(N,p)$ with $p\in\{0.10,0.20\}$,
        (ii) uniform random trees (Pr\"ufer sequences),
        (iii) Barab\'asi--Albert preferential attachment with $m\in\{2,3\}$,
        (iv) an unbalanced stochastic block model (two parts of sizes $\lfloor N/4\rfloor$ and $N-\lfloor N/4\rfloor$,
        within-part edge probability $0.05$ and cross-part probability $0.9$). We refer the reader to Appendix~\ref{sec:random-graph-ensembles} for more details.
        For each sampled graph $G$ we compute the adjacency spectrum and the overlaps $w_k$.

        We report the envelope ratio
        \begin{equation}\label{eq:ratio-metric}
        \mathrm{ratio}(G):=\frac{w_-(G)}
        {\displaystyle w_{-}{(S)}}
        ,
        \end{equation}
        with
        \begin{equation}
            w_{-}{(S)}=\frac{\big(\sqrt{N-1}-1\big)^2}{2},
        \end{equation}
        so that any graph $G$ with $\mathrm{ratio}(G)>1$ would be a challenger for the star's optimality.

\begin{figure}[t]
            \centering
            \includegraphics[width=0.47\textwidth]{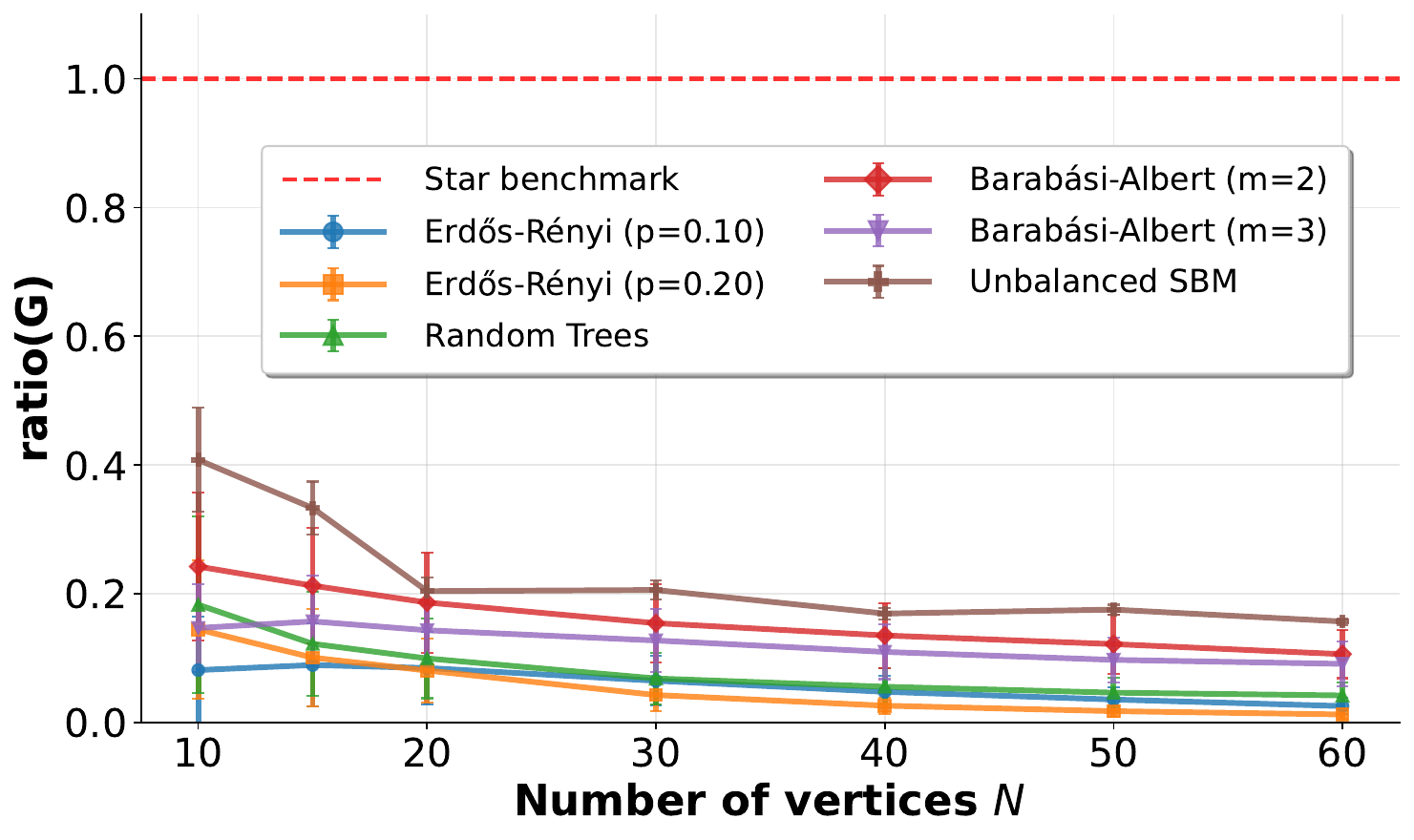}
            \caption{\textbf{Spectral envelope ratios for random graph families show no exceedances of the star benchmark.} Mean values and standard deviations of $\max_{\varepsilon_k<0}(\mathbf{1}^{\top} u_k)^2 / w_-^{(S)}$ across graph sizes $N \in \{10, 15, 20, 30, 40, 50, 60\}$ for six random graph models (500 samples per model per size). The red dashed line marks the star benchmark (ratio = 1). All tested architectures remain systematically below the benchmark, with ratios decreasing as system size increases, providing strong empirical evidence that no random topology surpasses the star's early-time charging performance.}
            \label{fig:scatter}
\end{figure}

\begin{figure}[t]
            \centering
            \includegraphics[width=0.47\textwidth]{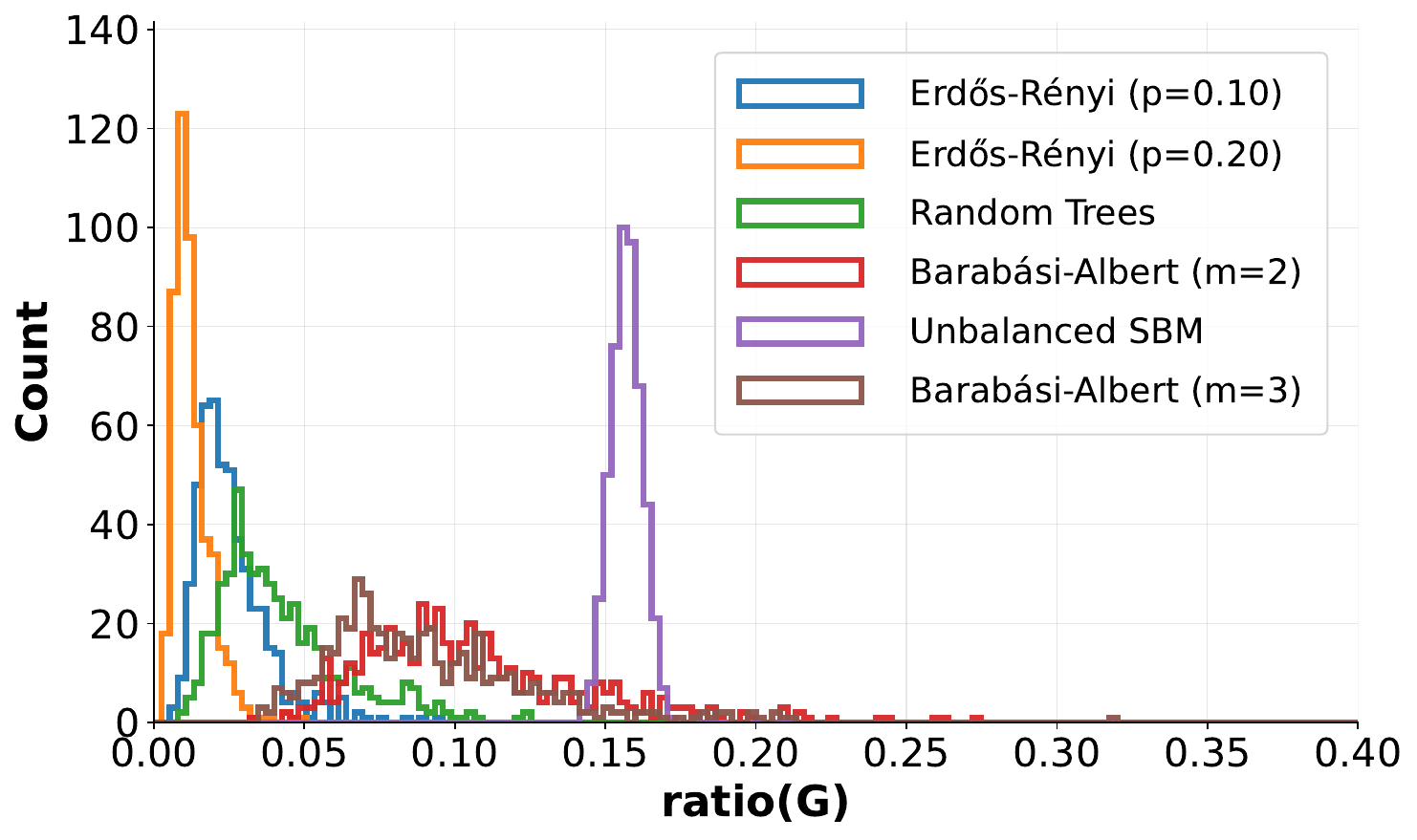}
            \caption{\textbf{Distribution analysis confirms systematic underperformance relative to the star benchmark.} Histograms of spectral envelope ratios $\max_{\varepsilon_k<0}(\mathbf{1}^{\top} u_k)^2 / w_-^{(S)}$ at the largest tested size ($N=60$) for six random graph families (500 samples each). Each panel shows the ratio distribution for: (top row) Erdős-Rényi graphs with edge probabilities $p=0.10$ and $p=0.20$; (middle row) uniform random trees and Barabási-Albert preferential attachment with $m=2$; (bottom row) unbalanced stochastic block model and Barabási-Albert with $m=3$. The red vertical line marks the star benchmark (ratio = 1). All distributions peak well below unity and exhibit minimal tail overlap with the benchmark, demonstrating that the star architecture's superiority is robust across diverse random topologies. Statistical summaries (max and mean values) are displayed for each family.}
            \label{fig:hists-by-model}
\end{figure}

        Across all sizes and all ensembles, we observe \emph{no exceedances}: the empirical ratios
        stay below $1$, typically with a visible margin that \emph{widens} as $N$ grows (see Fig.~\ref{fig:scatter}).
        Distributional views at $N=60$ (Fig.~\ref{fig:hists-by-model}) likewise show tight concentration strictly under
        the benchmark, with negligible tail mass near~$1$.

        These data align with the Sec.~\ref{subsec:h_non_trivial} mechanism.
        The envelope $\sup_h R_G(h)$ is driven by the largest uniform overlap in the negative
        sector. Random ER, tree, and preferential-attachment families almost never produce the
        strongly polarized, near-bipartite negative eigenvectors needed to accumulate a large
        uniform sum; even the deliberately unbalanced SBM, while closer in spirit, typically falls
        short of the star’s overlap. Combined with the analytical exclusions for regular and
        almost-regular constructions (Sec.~\ref{sec:expanders-no-star}), the empirical picture makes it highly
        implausible that generic sparse architectures can surpass the star at the envelope level.\footnote{This sweep is not an exhaustive search over \emph{all} $N$-vertex graphs; rather, it
        samples widely used random families and a biased bipartite model. Hence it provides
        evidence \emph{against} strong challengers in practice, complementing the formal bounds of Sec.~\ref{subsec:h_non_trivial}. We next turn to direct dynamical simulations of the
        microscopic model to confirm that the star’s envelope advantage translates into higher
        charging power under time evolution.}
    
\subsection{Direct simulation of the model}
Through this section we fix the parameters of the microscopic model to be $\kappa=\omega=1$ unless otherwise specified. The influence of these parameters on the system's dynamics is investigated in the Appendix~\ref{sec: micro params} and in the subsection~\ref{sec: Correlation}, and proves inconsequential for the star optimality.

\subsubsection{Power comparison across different graphs}

We directly simulated our model for all connected graphs with up to $N=7$ vertices and extracted corresponding values of $P_{\max}$ for $h=0$ and $=1$ (Fig.~\ref{fig:pmax_vs_N_h0_h1}). We confirmed that in both cases the star topology systematically yields the highest value of \(P_{\max}\) for each system size \(N\). The second best power is obtained by the ``perturbed star'', defined later (cf.~Fig.~\ref{fig:top3_structures}).
This is to be contrasted with the path graph which exhibits very low values of \(P_{\max}\) for \(N \geq 5\), and the complete graph which yields negligible charging power for all considered system sizes. 

\begin{figure}[h]
    \centering
    \begin{subfigure}[t]{0.48\textwidth}
        \centering
        \includegraphics[width=\textwidth]{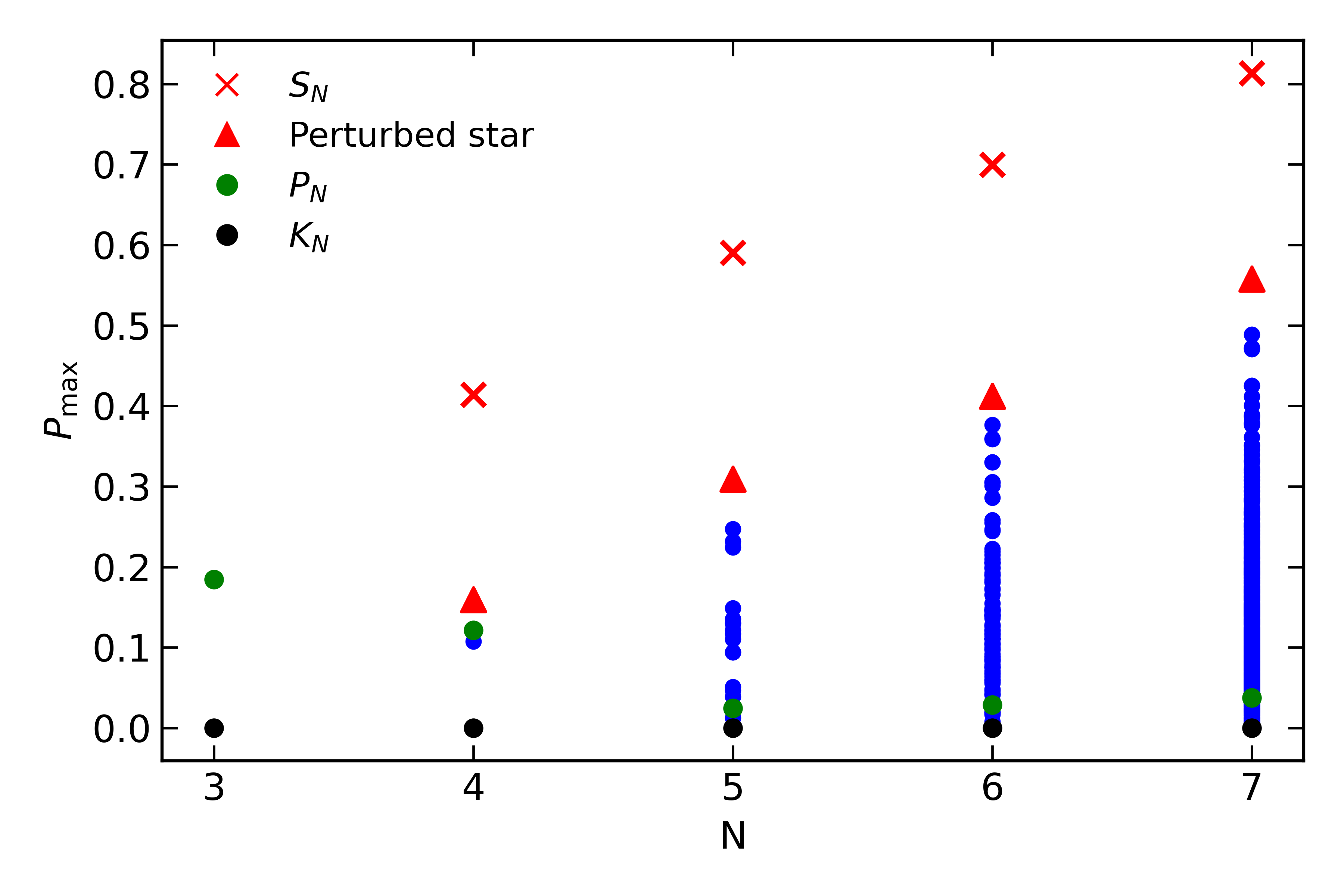}
        \label{fig:pmax_vs_N_h0}
    \end{subfigure}
    \hfill
    \begin{subfigure}[t]{0.48\textwidth}
        \centering
        \includegraphics[width=\textwidth]{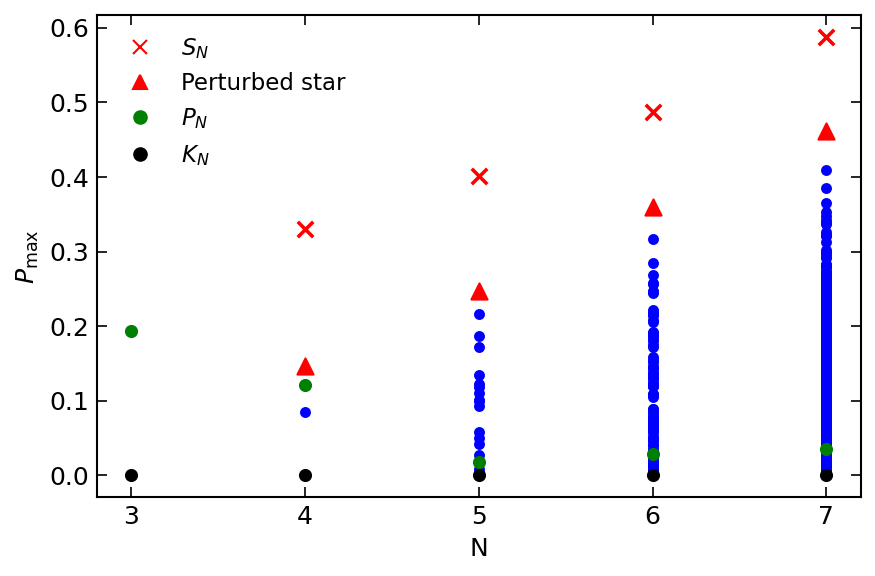}
        \label{fig:pmax_vs_N_h1}
    \end{subfigure}
    \caption{\textbf{$P_{\max}$ across graph topologies and system sizes.}
    Maximum charging power \(P_{\max}\) as a function of the number of sites $N$ for all connected interaction graphs, with $N=3,\dots,7$.
    Each blue dot corresponds to a distinct graph topology.
    The star graph ($S_N$), the perturbed star graph (second-best for \(N=7\); see Fig.~\ref{fig:top7-star-perturbed}), the path graph ($P_N$) and the complete graph  ($K_N$) are highlighted.
    The two panels correspond to different values of \(h\): \(h = 0\) (top) and \(h = 1\) (bottom), with \(\omega = \kappa = 1\).}
    \label{fig:pmax_vs_N_h0_h1}
\end{figure}

To further investigate the influence of battery's topology on the charging power we restrict ourselves to graphs on $N=7$ vertices. These can be indexed lexicographically according to their number of edges, degree sequence and number of automorphisms as implemented in the NetworkX graph atlas \cite{hagberg2008exploring}. As a result, increasing values of this index \(G\) roughly correspond to interaction graphs of increasing structural complexity and decreasing centralization. Figure~\ref{fig:pmax_vs_G} reveals a clear structural trend.
Observe that the data points are arranged into clear groups corresponding to different values of the number of edges $M$. Within each of such groups, the largest $P_{\max}$ is generally obtained by a graph with the lowest index, i.e.\ the one with the most star-like degree sequence. Moreover, the largest $P_{\max}$ within each group is generally decreasing with $M$ and so the largest $P_{\max}$ overall is obtained by the star graph itself.

We also look at the behaviour of $w_\textrm{min}(G)$ testing our hypothesis that it is the main factor governing the behaviour of $P_{\max}$. We confirm that the graphs with the largest $w_\textrm{min}(G)$ are these for which $P_{\max}$ is also maximal and the behaviour of both these quantities is qualitatively similar.

\begin{figure}[h!]
    \centering
    \begin{subfigure}[t]{0.48\textwidth}
        \centering
        \includegraphics[width=\textwidth]{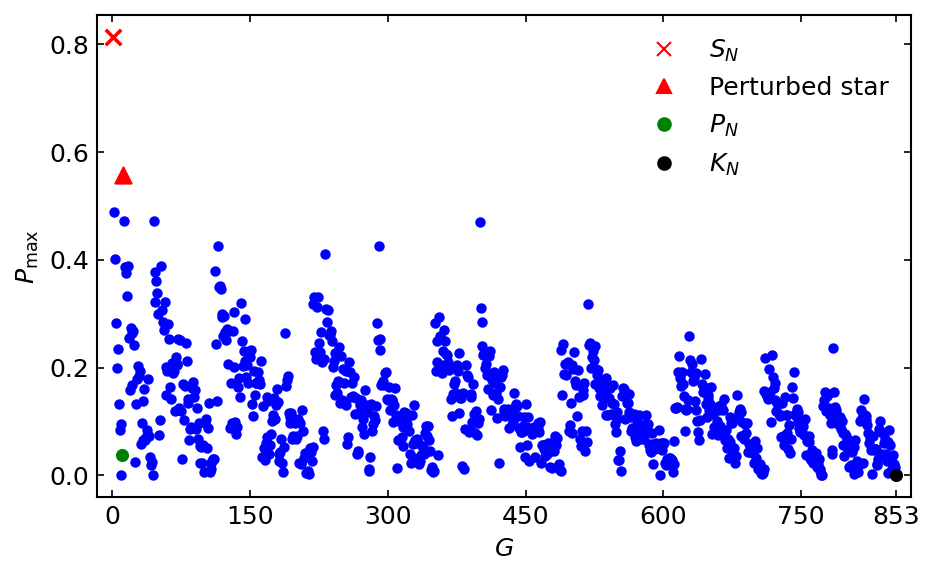}
        \label{fig:pmax_vs_G_h0}
    \end{subfigure}
    \hfill
    \begin{subfigure}[t]{0.48\textwidth}
        \centering
        \includegraphics[width=\textwidth]{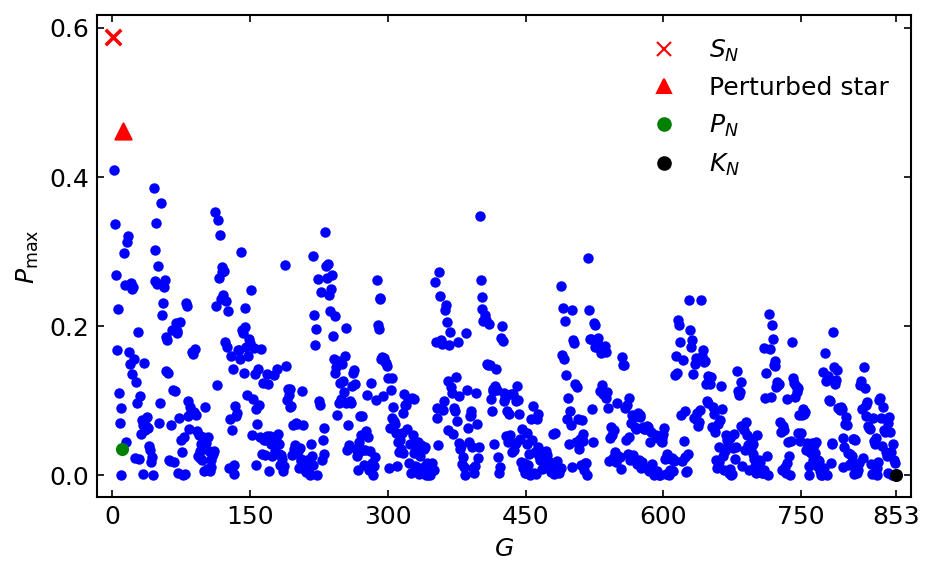}
        \label{fig:pmax_vs_G_h1}
    \end{subfigure}
    \hfill
    \begin{subfigure}[t]{0.49\textwidth}
        \centering
        \includegraphics[width=\textwidth]{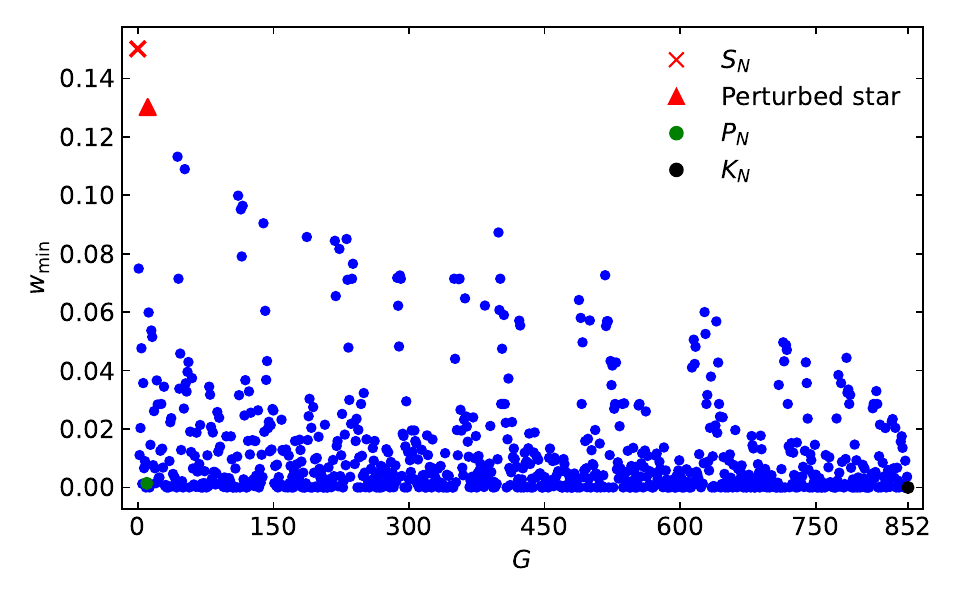}
        \label{fig:wmin_vs_G_h1}
    \end{subfigure}
    \caption{\textbf{$P_{\max}$ and $w_\mathrm{min}$ as a function of the topological index at fixed system size.} $N=7$ for all connected interaction graphs. Canonical graph topologies are highlighted: star ($S_N$), perturbed star, path ($P_N$), and complete graph ($K_N$).
    The index \(G\) follows the canonical ordering of the NetworkX graph atlas. The first plot corresponds to $h=0$, the second to $h=1$. The plot of $w_\mathrm{min}$ is for $h=0$.
    }
    \label{fig:pmax_vs_G}
\end{figure}

Finally, we investigate the stability of the star topology. Figure~\ref{fig:top3_structures} displays the three interaction graphs that achieve the highest values of \(P_{\max}\).
All three exhibit a strongly centralized, star-like architecture dominated by a single hub, which concentrates most of the couplings.

The optimal configuration is the pure star shown in panel (a), where all six edges are attached to a single central node.
The second structure, panel (b), is a minimal perturbation of the star obtained by adding a single extra edge between two peripheral nodes, introducing a local motif while preserving the dominant hub. We call such graphs ``perturbed stars''.
Finally, panel (c) corresponds to a more asymmetric star variant in which the connectivity is slightly redistributed: one peripheral node is linked to another leaf, forming a short chain attached to the hub and breaking the perfect single-hub symmetry.
\begin{figure}[h!]
    \centering
    \begin{subfigure}{0.32\linewidth}
        \centering
        \includegraphics[width=\linewidth]{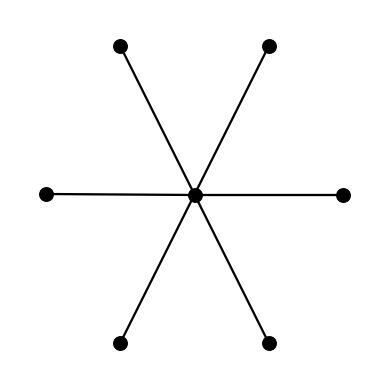}
        \subcaption{$P_{\max} \approx 0.81$}
        \label{fig:top7-star}
    \end{subfigure}
    \hfill
    \begin{subfigure}{0.32\linewidth}
        \centering
        \includegraphics[width=\linewidth]{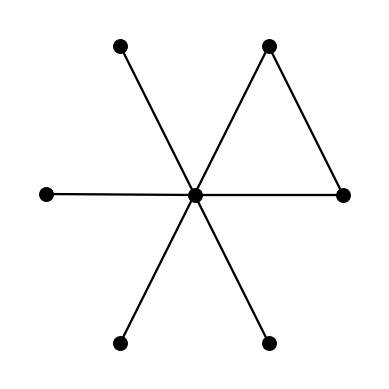}
        \subcaption{$P_{\max} \approx 0.55$}
        \label{fig:top7-star-perturbed}
    \end{subfigure}
    \hfill
    \begin{subfigure}{0.32\linewidth}
        \centering
        \includegraphics[width=\linewidth]{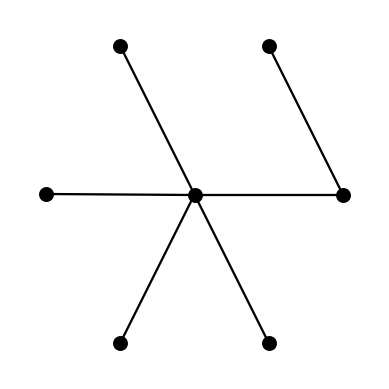}
        \subcaption{$P_{\max} \approx 0.48$}
        \label{fig:top7-asym}
    \end{subfigure}

    \caption{\textbf{Top three interaction topologies for \( N = 7 \).}
    }
    \label{fig:top3_structures}
\end{figure}

Centralized star-like structures thus emerge as the optimal design principle: concentrating interactions around a single hub is the most effective way to maximize the charging power.

\subsubsection{Scaling of maximum average power with system size}

A commonly discussed signature of cooperative quantum effects in charging protocols is a superlinear scaling of the charging power with system size \cite{campaioli2017enhancing}.
To probe this behavior in our model, we focus on the star topology and study how the maximum average charging power scales with the system size.

For each system size \(N \in \{4,\dots,11\}\), we simulate the charging dynamics of the star graph and extract \(P_{\max}\).
The resulting values \(P_{\max}(N)\) are then fitted to a power-law scaling
\begin{equation}
    P_{\max}(N) = a\,N^{\eta},
\end{equation}
from which the exponent \(\eta\) is obtained numerically.

As shown in Fig.~\ref{fig:scaling_star}, we observe a clear superlinear scaling for the star architecture, with exponent
\(\eta \simeq 1.18\).
This indicates the emergence of strong collective effects in the charging dynamics as the system size increases.
Such behavior is consistent with the physical picture that concentrating interactions around a single hub enhances cooperative energy transfer and leads to an amplification of the achievable peak charging power.

\begin{figure}[h!]
    \centering
    \includegraphics[width=0.47\textwidth]{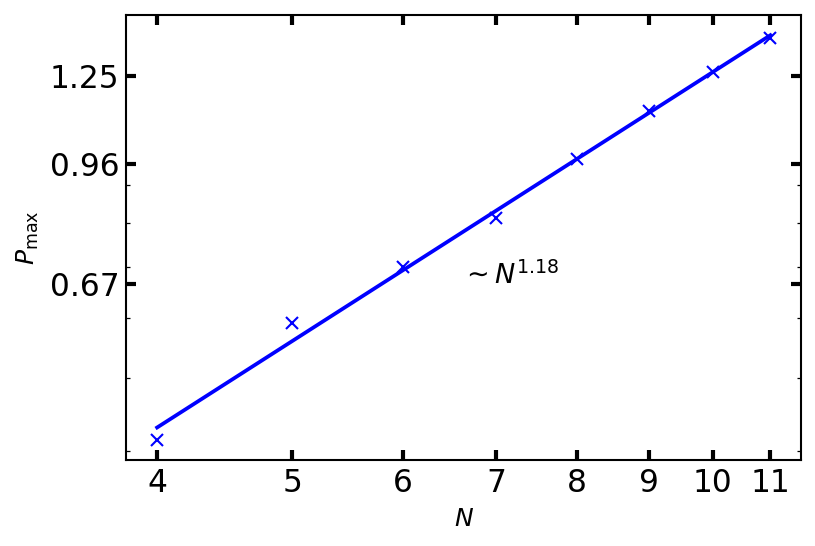}
    \caption{\textbf{Scaling of the maximum average charging power for the star topology.}
    For each system size $N$, the maximum average power
    \(P_{\max}(N)\) is extracted from the charging dynamics and fitted to a power law
    \(P_{\max}(N)\sim N^{\eta}\), yielding \(\eta= 1.18\pm0.03\).
 }
    \label{fig:scaling_star}
\end{figure}

\subsubsection{Correlation between initial and maximal charging power}\label{sec: Correlation}

Finally, we investigate the relationship between the short-time charging behaviour and the overall charging performance of the quantum battery.
To this end, we consider all interaction graphs with $N=7$ and compare the initial charging power $P_{\mathrm{init}}$, defined as the power after the first time step of the dynamics $dt = 0.01$, with the maximal average power $P_{\max}$ attained during the full time evolution.
In order to assess the robustness of this relation, we explore a broad range of parameter regimes by varying the driving frequency $\omega$ and the interaction strength $\kappa$ over the set \(\omega,\kappa \in \{0.1,1,5\}\), thereby spanning several orders of magnitude.

Figure~\ref{fig:PmaxPinit_grid} displays scatter plots of $P_{\mathrm{init}}$ versus $P_{\max}$ for the complete set of $N=7$ graph topologies.
The figure is divided into two panels corresponding to the two cases $h=0$ and $h=1$.
Within each panel, different subplots explore the various combinations of $(\omega,\kappa)$ arranged on a grid.

In all considered parameter regimes, a clear positive correlation between the initial and maximal charging powers is observed.
Graph topologies that exhibit a larger initial charging power systematically reach a higher maximal average power, independently of the specific choice of parameters.
This trend is quantitatively confirmed by the Pearson correlation coefficient, which takes values in the range \(r_P \in [0.81,1.00]\), depending on the parameters $(\omega,\kappa)$, and for both values of $h$. The robustness of this correlation across widely different parameter regimes demonstrates that the relation between early-time and long-time charging performance is essentially independent of the microscopic details of the dynamics.
The short-time charging behavior already encodes most of the relevant structural information of the interaction topology, even in regimes where the dynamics becomes strongly nonlinear.

From a practical perspective, this result shows that the initial charging power $P_{\mathrm{init}}$ can be used as a computationally efficient and reliable proxy for the maximal average charging power $P_{\max}$.

\begin{figure}[h!]
    \centering
    \includegraphics[width=1\linewidth]{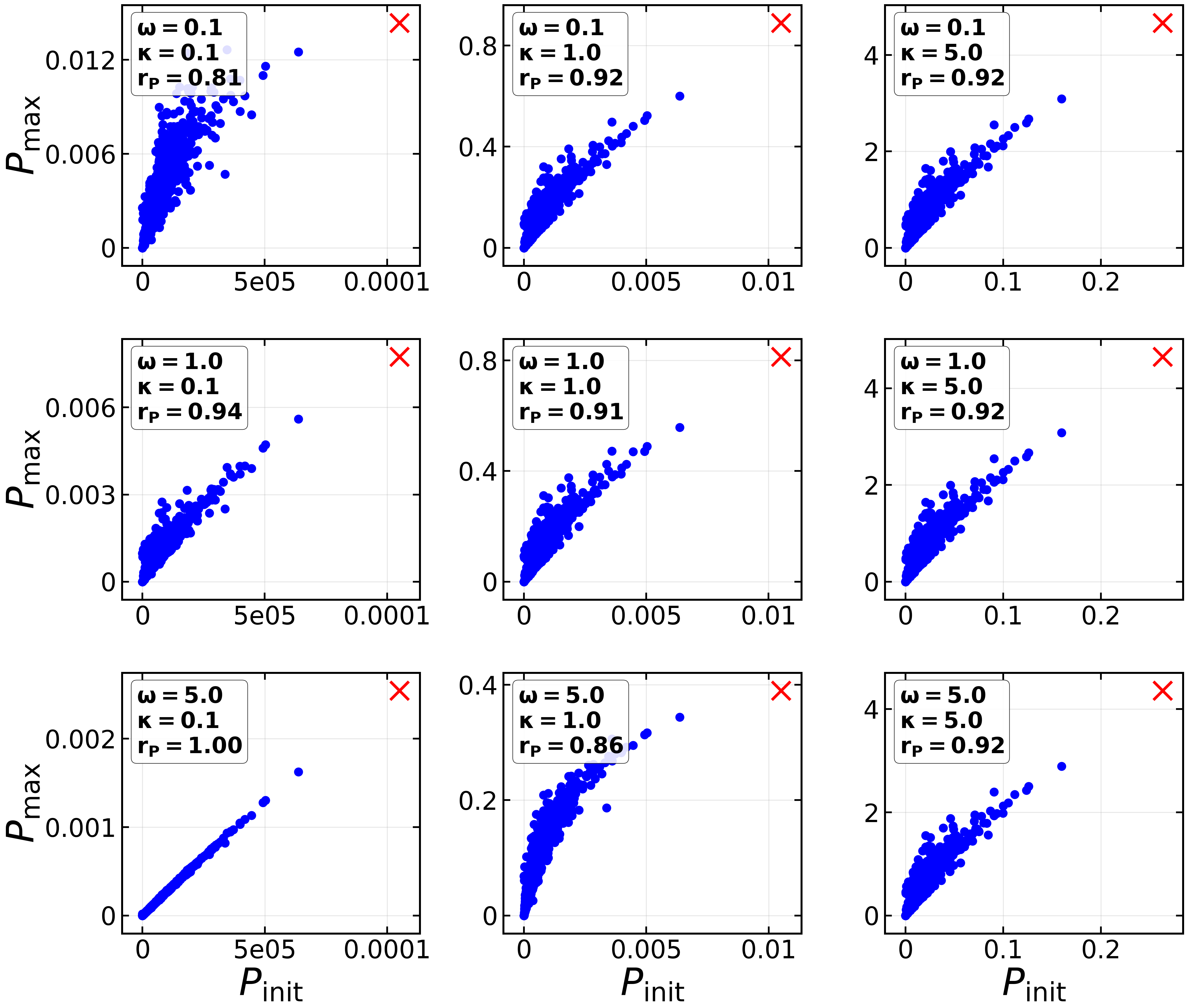}\\
    \vspace{0.5cm}
    \includegraphics[width=1.0\linewidth]{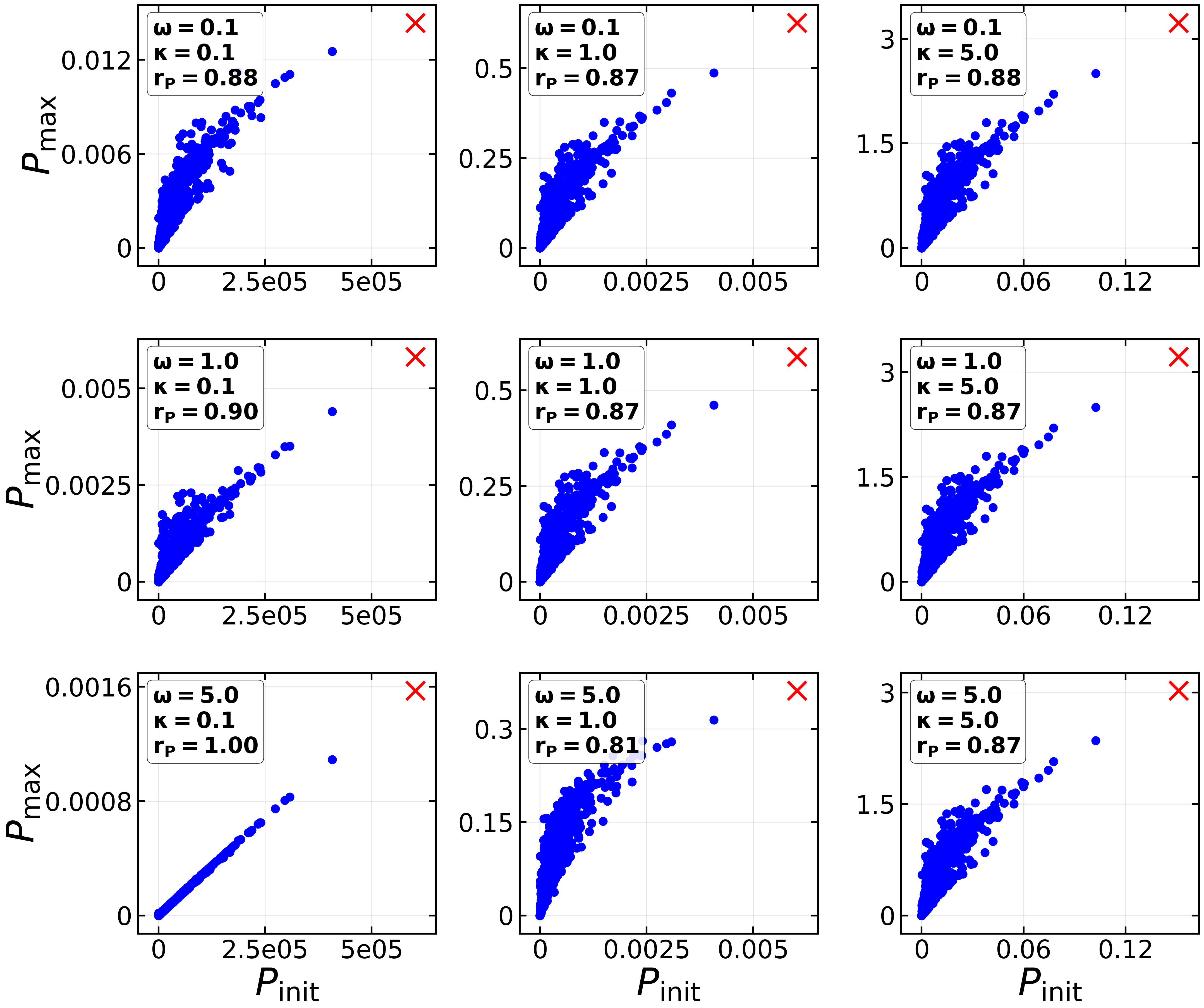}
    \caption{\textbf{Initial charging power $P_{\mathrm{init}}$ versus maximal average charging power $P_{\max}$ for all interaction graphs with $N=7$.}
    Top: $h=0$. Bottom: $h=1$.
    Each panel displays a grid of subplots corresponding to different combinations of the parameters $(\omega,\kappa) \in \{0.1,1,5\}$.
    In all cases, a strong positive correlation is observed between $P_{\mathrm{init}}$ and $P_{\max}$, with Pearson coefficients ranging from $r_P \simeq 0.81$ up to $r_P = 1.00$.
    This demonstrates the robustness and parameter-independence of the correlation between early-time and maximal charging performance. The ``\textcolor{red}{$\times$}'' marker corresponds to the star graph. }
    \label{fig:PmaxPinit_grid}
\end{figure}

\section{Conclusions and Outlook}
\label{sec:conclusion}

    We have framed architecture design for quantum batteries within a spectral graph--theoretic lens and proved that \emph{star} graphs maximize physically relevant bounds to instantaneous charging power. At a mechanistic level, the hub-and-spoke layout concentrates coupling through the principal eigenmode, thereby boosting the spectral radius and the eigenvector-weighted sums that control our protocol's performance. We complemented these bounds with explicit implementations and numerical comparisons, establishing a clear and scalable advantage for stars over other architectures, like complete or path graphs at fixed number of edges.

    Beyond delivering a compact optimality statement, our results suggest concrete principles for near-term devices. Whenever interactions are sparse, prioritizing a single mediating hub---as naturally realized in cavity and circuit QED, trapped ions, and resonator-coupled platforms---yields the largest power for a given hardware budget.

    Several extensions are both natural and valuable.
    (i) \emph{Locality and hardware constraints.} In spatially embedded architectures with distance-dependent couplings, one may trade a single hub for a small number of regional hubs; determining optimal multi-hub topologies as a function of geometric and fabrication constraints is an open problem.
    (ii) \emph{Weighted and time-dependent graphs.} Allowing edge weights and control schedules invites optimal-control formulations on graphs; extending our bounds to weighted Laplacians and to periodic driving is a concrete direction.
    (iii) \emph{Alternative quality metrics.} Power is only one dimension in the space of reasonable quality metrics for quantum batteries: asymptotic extractable work, variance, and cycle-to-cycle repeatability matter for practical loads. Bounding joint trade-offs between power and stability is an important challenge.
    (iv) \emph{Open-system resources.} Correlated noise, non-Markovian baths, and measurement-assisted or feedback-enhanced protocols can possibly reshape optimal topologies; integrating these resources into the graph framework could reveal design rules that go beyond coherent unitary charging.
    (v) \emph{Resource accounting.} A complete assessment should include the control overhead needed to realize a hub, state-preparation costs for the charger, etc... incorporating such costs into constrained graph optimization is a promising avenue.
    (vi) \emph{Discharging and energy transfer.} The same spectral mechanisms that boost charging are relevant to rapid and targeted energy delivery. Extending the theory to optimal discharging and routing on networks is directly relevant for concrete QB as embedded in an actual device.

    We expect the graph perspective developed here to serve as a compact interface between theory and experiment, helping to translate hardware constraints into provably optimal layouts. In the broader roadmap of quantum batteries~\cite{campaioli2024colloquium}, our results clarify how topology alone can underwrite collective power gains and point to concrete routes for scaling sparse, power-efficient devices.

\section*{Acknowledgments}

    A.T. acknowledges funding via the FNR-CORE Grant ``BroadApp'' (FNR-CORE C20/MS/14769845) and ERC-AdG Grant ``FITMOL''.

    This research was funded in whole, or in part, by the Luxembourg National Research
    Fund (FNR), grant reference C24/MS/18969168/MOS. For the purpose of open access, and
    in fulfilment of the obligations arising from the grant agreement, the author has applied
    a Creative Commons Attribution 4.0 International (CC BY 4.0) license to any Author
    Accepted Manuscript version arising from this submission.

\section*{Authors declarations}

    \subsection*{Conflict of Interest}

        The authors have no conflicts to disclose.

    \subsection*{Author Contributions}

        \begin{itemize}
            \item \textbf{Matthieu Sarkis}: Conceptualization, Formal analysis (analytical and numerical), Writing—original draft and editing, Supervision.
            \item \textbf{Oskar A. Pro\'sniak}: Formal analysis (analytical and numerical), Writing—original draft and editing.
            \item \textbf{Samuel Nigro}: Formal analysis (numerical), Writing—original draft and editing.
            \item \textbf{Alexandre Tkatchenko}: Supervision, Funding acquisition, Writing—review and editing.
        \end{itemize}

\section*{Data availability}
    The data generated during this work can be provided upon reasonable request.

\bibliographystyle{apsrev4-2}
\bibliography{references}

\begin{thebibliography}{43}%
\makeatletter
\providecommand \@ifxundefined [1]{%
 \@ifx{#1\undefined}
}%
\providecommand \@ifnum [1]{%
 \ifnum #1\expandafter \@firstoftwo
 \else \expandafter \@secondoftwo
 \fi
}%
\providecommand \@ifx [1]{%
 \ifx #1\expandafter \@firstoftwo
 \else \expandafter \@secondoftwo
 \fi
}%
\providecommand \natexlab [1]{#1}%
\providecommand \enquote  [1]{``#1''}%
\providecommand \bibnamefont  [1]{#1}%
\providecommand \bibfnamefont [1]{#1}%
\providecommand \citenamefont [1]{#1}%
\providecommand \href@noop [0]{\@secondoftwo}%
\providecommand \href [0]{\begingroup \@sanitize@url \@href}%
\providecommand \@href[1]{\@@startlink{#1}\@@href}%
\providecommand \@@href[1]{\endgroup#1\@@endlink}%
\providecommand \@sanitize@url [0]{\catcode `\\12\catcode `\$12\catcode
  `\&12\catcode `\#12\catcode `\^12\catcode `\_12\catcode `\%12\relax}%
\providecommand \@@startlink[1]{}%
\providecommand \@@endlink[0]{}%
\providecommand \url  [0]{\begingroup\@sanitize@url \@url }%
\providecommand \@url [1]{\endgroup\@href {#1}{\urlprefix }}%
\providecommand \urlprefix  [0]{URL }%
\providecommand \Eprint [0]{\href }%
\providecommand \doibase [0]{https://doi.org/}%
\providecommand \selectlanguage [0]{\@gobble}%
\providecommand \bibinfo  [0]{\@secondoftwo}%
\providecommand \bibfield  [0]{\@secondoftwo}%
\providecommand \translation [1]{[#1]}%
\providecommand \BibitemOpen [0]{}%
\providecommand \bibitemStop [0]{}%
\providecommand \bibitemNoStop [0]{.\EOS\space}%
\providecommand \EOS [0]{\spacefactor3000\relax}%
\providecommand \BibitemShut  [1]{\csname bibitem#1\endcsname}%
\let\auto@bib@innerbib\@empty
\bibitem [{\citenamefont {Alicki}\ and\ \citenamefont
  {Fannes}(2013)}]{alicki2013entanglement}%
  \BibitemOpen
  \bibfield  {author} {\bibinfo {author} {\bibfnamefont {R.}~\bibnamefont
  {Alicki}}\ and\ \bibinfo {author} {\bibfnamefont {M.}~\bibnamefont
  {Fannes}},\ }\href {https://doi.org/10.1103/PhysRevE.87.042123} {\bibfield
  {journal} {\bibinfo  {journal} {Physical Review E—Statistical, Nonlinear,
  and Soft Matter Physics}\ }\textbf {\bibinfo {volume} {87}},\ \bibinfo
  {pages} {042123} (\bibinfo {year} {2013})}\BibitemShut {NoStop}%
\bibitem [{\citenamefont {Ferraro}\ \emph {et~al.}(2018)\citenamefont
  {Ferraro}, \citenamefont {Campisi}, \citenamefont {Andolina}, \citenamefont
  {Pellegrini},\ and\ \citenamefont {Polini}}]{ferraro2018highpower}%
  \BibitemOpen
  \bibfield  {author} {\bibinfo {author} {\bibfnamefont {D.}~\bibnamefont
  {Ferraro}}, \bibinfo {author} {\bibfnamefont {M.}~\bibnamefont {Campisi}},
  \bibinfo {author} {\bibfnamefont {G.~M.}\ \bibnamefont {Andolina}}, \bibinfo
  {author} {\bibfnamefont {V.}~\bibnamefont {Pellegrini}},\ and\ \bibinfo
  {author} {\bibfnamefont {M.}~\bibnamefont {Polini}},\ }\href
  {https://doi.org/10.1103/PhysRevLett.120.117702} {\bibfield  {journal}
  {\bibinfo  {journal} {Physical Review Letters}\ }\textbf {\bibinfo {volume}
  {120}},\ \bibinfo {pages} {117702} (\bibinfo {year} {2018})}\BibitemShut
  {NoStop}%
\bibitem [{\citenamefont {Andolina}\ \emph
  {et~al.}(2019{\natexlab{a}})\citenamefont {Andolina}, \citenamefont {Keck},
  \citenamefont {Mari}, \citenamefont {Campisi}, \citenamefont {Giovannetti},\
  and\ \citenamefont {Polini}}]{andolina2019extractable}%
  \BibitemOpen
  \bibfield  {author} {\bibinfo {author} {\bibfnamefont {G.~M.}\ \bibnamefont
  {Andolina}}, \bibinfo {author} {\bibfnamefont {M.}~\bibnamefont {Keck}},
  \bibinfo {author} {\bibfnamefont {A.}~\bibnamefont {Mari}}, \bibinfo {author}
  {\bibfnamefont {M.}~\bibnamefont {Campisi}}, \bibinfo {author} {\bibfnamefont
  {V.}~\bibnamefont {Giovannetti}},\ and\ \bibinfo {author} {\bibfnamefont
  {M.}~\bibnamefont {Polini}},\ }\href
  {https://doi.org/10.1103/PhysRevLett.122.047702} {\bibfield  {journal}
  {\bibinfo  {journal} {Physical Review Letters}\ }\textbf {\bibinfo {volume}
  {122}},\ \bibinfo {pages} {047702} (\bibinfo {year}
  {2019}{\natexlab{a}})}\BibitemShut {NoStop}%
\bibitem [{\citenamefont {Campaioli}\ \emph {et~al.}(2024)\citenamefont
  {Campaioli}, \citenamefont {Gherardini}, \citenamefont {Quach}, \citenamefont
  {Polini},\ and\ \citenamefont {Andolina}}]{campaioli2024colloquium}%
  \BibitemOpen
  \bibfield  {author} {\bibinfo {author} {\bibfnamefont {F.}~\bibnamefont
  {Campaioli}}, \bibinfo {author} {\bibfnamefont {S.}~\bibnamefont
  {Gherardini}}, \bibinfo {author} {\bibfnamefont {J.~Q.}\ \bibnamefont
  {Quach}}, \bibinfo {author} {\bibfnamefont {M.}~\bibnamefont {Polini}},\ and\
  \bibinfo {author} {\bibfnamefont {G.~M.}\ \bibnamefont {Andolina}},\
  }\href@noop {} {\bibfield  {journal} {\bibinfo  {journal} {Reviews of Modern
  Physics}\ }\textbf {\bibinfo {volume} {96}},\ \bibinfo {pages} {031001}
  (\bibinfo {year} {2024})}\BibitemShut {NoStop}%
\bibitem [{\citenamefont {Friis}\ and\ \citenamefont
  {Huber}(2018)}]{friis2018gaussian}%
  \BibitemOpen
  \bibfield  {author} {\bibinfo {author} {\bibfnamefont {N.}~\bibnamefont
  {Friis}}\ and\ \bibinfo {author} {\bibfnamefont {M.}~\bibnamefont {Huber}},\
  }\href {https://doi.org/10.22331/q-2018-04-23-61} {\bibfield  {journal}
  {\bibinfo  {journal} {Quantum}\ }\textbf {\bibinfo {volume} {2}},\ \bibinfo
  {pages} {61} (\bibinfo {year} {2018})}\BibitemShut {NoStop}%
\bibitem [{\citenamefont {Barra}(2019)}]{barra2019dissipative}%
  \BibitemOpen
  \bibfield  {author} {\bibinfo {author} {\bibfnamefont {F.}~\bibnamefont
  {Barra}},\ }\href {https://doi.org/10.1103/PhysRevLett.122.210601} {\bibfield
   {journal} {\bibinfo  {journal} {Physical Review Letters}\ }\textbf {\bibinfo
  {volume} {122}},\ \bibinfo {pages} {210601} (\bibinfo {year}
  {2019})}\BibitemShut {NoStop}%
\bibitem [{\citenamefont {Campaioli}\ \emph {et~al.}(2017)\citenamefont
  {Campaioli}, \citenamefont {Pollock}, \citenamefont {Binder}, \citenamefont
  {C{'e}leri}, \citenamefont {Goold}, \citenamefont {Vinjanampathy},\ and\
  \citenamefont {Modi}}]{campaioli2017enhancing}%
  \BibitemOpen
  \bibfield  {author} {\bibinfo {author} {\bibfnamefont {F.}~\bibnamefont
  {Campaioli}}, \bibinfo {author} {\bibfnamefont {F.~A.}\ \bibnamefont
  {Pollock}}, \bibinfo {author} {\bibfnamefont {F.~C.}\ \bibnamefont {Binder}},
  \bibinfo {author} {\bibfnamefont {L.}~\bibnamefont {C{'e}leri}}, \bibinfo
  {author} {\bibfnamefont {J.}~\bibnamefont {Goold}}, \bibinfo {author}
  {\bibfnamefont {S.}~\bibnamefont {Vinjanampathy}},\ and\ \bibinfo {author}
  {\bibfnamefont {K.}~\bibnamefont {Modi}},\ }\href@noop {} {\bibfield
  {journal} {\bibinfo  {journal} {Physical review letters}\ }\textbf {\bibinfo
  {volume} {118}},\ \bibinfo {pages} {150601} (\bibinfo {year}
  {2017})}\BibitemShut {NoStop}%
\bibitem [{\citenamefont {Julia-Farre}\ \emph {et~al.}(2020)\citenamefont
  {Julia-Farre}, \citenamefont {Salamon}, \citenamefont {Riera}, \citenamefont
  {Bera},\ and\ \citenamefont {Lewenstein}}]{juliafarre2020bounds}%
  \BibitemOpen
  \bibfield  {author} {\bibinfo {author} {\bibfnamefont {S.}~\bibnamefont
  {Julia-Farre}}, \bibinfo {author} {\bibfnamefont {T.}~\bibnamefont
  {Salamon}}, \bibinfo {author} {\bibfnamefont {A.}~\bibnamefont {Riera}},
  \bibinfo {author} {\bibfnamefont {M.~N.}\ \bibnamefont {Bera}},\ and\
  \bibinfo {author} {\bibfnamefont {M.}~\bibnamefont {Lewenstein}},\ }\href
  {https://doi.org/10.1103/PhysRevResearch.2.023113} {\bibfield  {journal}
  {\bibinfo  {journal} {Physical Review Research}\ }\textbf {\bibinfo {volume}
  {2}},\ \bibinfo {pages} {023113} (\bibinfo {year} {2020})}\BibitemShut
  {NoStop}%
\bibitem [{\citenamefont {Gyhm}\ \emph
  {et~al.}(2022{\natexlab{a}})\citenamefont {Gyhm}, \citenamefont
  {{\v{S}}afr{\'a}nek},\ and\ \citenamefont {Rosa}}]{gyhm2021advantage}%
  \BibitemOpen
  \bibfield  {author} {\bibinfo {author} {\bibfnamefont {J.-Y.}\ \bibnamefont
  {Gyhm}}, \bibinfo {author} {\bibfnamefont {D.}~\bibnamefont
  {{\v{S}}afr{\'a}nek}},\ and\ \bibinfo {author} {\bibfnamefont
  {D.}~\bibnamefont {Rosa}},\ }\href
  {https://doi.org/10.1103/PhysRevLett.128.140501} {\bibfield  {journal}
  {\bibinfo  {journal} {Physical Review Letters}\ }\textbf {\bibinfo {volume}
  {128}},\ \bibinfo {pages} {140501} (\bibinfo {year}
  {2022}{\natexlab{a}})}\BibitemShut {NoStop}%
\bibitem [{\citenamefont {Gyhm}\ \emph
  {et~al.}(2022{\natexlab{b}})\citenamefont {Gyhm}, \citenamefont
  {{\v{S}}afr{\'a}nek},\ and\ \citenamefont {Rosa}}]{gyhm2022cannot}%
  \BibitemOpen
  \bibfield  {author} {\bibinfo {author} {\bibfnamefont {J.-Y.}\ \bibnamefont
  {Gyhm}}, \bibinfo {author} {\bibfnamefont {D.}~\bibnamefont
  {{\v{S}}afr{\'a}nek}},\ and\ \bibinfo {author} {\bibfnamefont
  {D.}~\bibnamefont {Rosa}},\ }\href
  {https://doi.org/10.1103/PhysRevLett.128.140501} {\bibfield  {journal}
  {\bibinfo  {journal} {Phys. Rev. Lett.}\ }\textbf {\bibinfo {volume} {128}},\
  \bibinfo {pages} {140501} (\bibinfo {year} {2022}{\natexlab{b}})},\ \Eprint
  {https://arxiv.org/abs/2108.02491} {arXiv:2108.02491 [quant-ph]} \BibitemShut
  {NoStop}%
\bibitem [{\citenamefont {Gyhm}\ \emph {et~al.}(2024)\citenamefont {Gyhm},
  \citenamefont {Rosa},\ and\ \citenamefont
  {{\v{S}}afr{\'a}nek}}]{gyhm2024minimal}%
  \BibitemOpen
  \bibfield  {author} {\bibinfo {author} {\bibfnamefont {J.-Y.}\ \bibnamefont
  {Gyhm}}, \bibinfo {author} {\bibfnamefont {D.}~\bibnamefont {Rosa}},\ and\
  \bibinfo {author} {\bibfnamefont {D.}~\bibnamefont {{\v{S}}afr{\'a}nek}},\
  }\href {https://doi.org/10.1103/PhysRevA.109.022607} {\bibfield  {journal}
  {\bibinfo  {journal} {Phys. Rev. A}\ }\textbf {\bibinfo {volume} {109}},\
  \bibinfo {pages} {022607} (\bibinfo {year} {2024})},\ \Eprint
  {https://arxiv.org/abs/2308.16086} {arXiv:2308.16086 [quant-ph]} \BibitemShut
  {NoStop}%
\bibitem [{\citenamefont {Le}\ \emph {et~al.}(2018)\citenamefont {Le},
  \citenamefont {Levinsen}, \citenamefont {Modi}, \citenamefont {Parish},\ and\
  \citenamefont {Pollock}}]{le2018spin}%
  \BibitemOpen
  \bibfield  {author} {\bibinfo {author} {\bibfnamefont {T.~P.}\ \bibnamefont
  {Le}}, \bibinfo {author} {\bibfnamefont {J.}~\bibnamefont {Levinsen}},
  \bibinfo {author} {\bibfnamefont {K.}~\bibnamefont {Modi}}, \bibinfo {author}
  {\bibfnamefont {M.~M.}\ \bibnamefont {Parish}},\ and\ \bibinfo {author}
  {\bibfnamefont {F.~A.}\ \bibnamefont {Pollock}},\ }\href@noop {} {\bibfield
  {journal} {\bibinfo  {journal} {Physical Review A}\ }\textbf {\bibinfo
  {volume} {97}},\ \bibinfo {pages} {022106} (\bibinfo {year}
  {2018})}\BibitemShut {NoStop}%
\bibitem [{\citenamefont {Andolina}\ \emph
  {et~al.}(2019{\natexlab{b}})\citenamefont {Andolina}, \citenamefont {Keck},
  \citenamefont {Mari}, \citenamefont {Polini},\ and\ \citenamefont
  {Giovannetti}}]{andolina2019versus}%
  \BibitemOpen
  \bibfield  {author} {\bibinfo {author} {\bibfnamefont {G.~M.}\ \bibnamefont
  {Andolina}}, \bibinfo {author} {\bibfnamefont {M.}~\bibnamefont {Keck}},
  \bibinfo {author} {\bibfnamefont {A.}~\bibnamefont {Mari}}, \bibinfo {author}
  {\bibfnamefont {M.}~\bibnamefont {Polini}},\ and\ \bibinfo {author}
  {\bibfnamefont {V.}~\bibnamefont {Giovannetti}},\ }\href
  {https://doi.org/10.1103/PhysRevB.99.205437} {\bibfield  {journal} {\bibinfo
  {journal} {Physical Review B}\ }\textbf {\bibinfo {volume} {99}},\ \bibinfo
  {pages} {205437} (\bibinfo {year} {2019}{\natexlab{b}})}\BibitemShut
  {NoStop}%
\bibitem [{\citenamefont {Rossini}\ \emph {et~al.}(2019)\citenamefont
  {Rossini}, \citenamefont {Andolina},\ and\ \citenamefont
  {Polini}}]{rossini2019mbl}%
  \BibitemOpen
  \bibfield  {author} {\bibinfo {author} {\bibfnamefont {D.}~\bibnamefont
  {Rossini}}, \bibinfo {author} {\bibfnamefont {G.~M.}\ \bibnamefont
  {Andolina}},\ and\ \bibinfo {author} {\bibfnamefont {M.}~\bibnamefont
  {Polini}},\ }\href {https://doi.org/10.1103/PhysRevB.100.115142} {\bibfield
  {journal} {\bibinfo  {journal} {Physical Review B}\ }\textbf {\bibinfo
  {volume} {100}},\ \bibinfo {pages} {115142} (\bibinfo {year}
  {2019})}\BibitemShut {NoStop}%
\bibitem [{\citenamefont {Rossini}\ \emph
  {et~al.}(2020{\natexlab{a}})\citenamefont {Rossini}, \citenamefont
  {Andolina}, \citenamefont {Rosa}, \citenamefont {Carrega},\ and\
  \citenamefont {Polini}}]{rossini2020syk}%
  \BibitemOpen
  \bibfield  {author} {\bibinfo {author} {\bibfnamefont {D.}~\bibnamefont
  {Rossini}}, \bibinfo {author} {\bibfnamefont {G.~M.}\ \bibnamefont
  {Andolina}}, \bibinfo {author} {\bibfnamefont {D.}~\bibnamefont {Rosa}},
  \bibinfo {author} {\bibfnamefont {M.}~\bibnamefont {Carrega}},\ and\ \bibinfo
  {author} {\bibfnamefont {M.}~\bibnamefont {Polini}},\ }\href
  {https://doi.org/10.1103/PhysRevLett.125.236402} {\bibfield  {journal}
  {\bibinfo  {journal} {Physical Review Letters}\ }\textbf {\bibinfo {volume}
  {125}},\ \bibinfo {pages} {236402} (\bibinfo {year}
  {2020}{\natexlab{a}})}\BibitemShut {NoStop}%
\bibitem [{\citenamefont {Carrasco}\ \emph {et~al.}(2022)\citenamefont
  {Carrasco}, \citenamefont {Maze}, \citenamefont {Hermann-Avigliano},\ and\
  \citenamefont {Barra}}]{carrasco2022collective}%
  \BibitemOpen
  \bibfield  {author} {\bibinfo {author} {\bibfnamefont {J.}~\bibnamefont
  {Carrasco}}, \bibinfo {author} {\bibfnamefont {J.~R.}\ \bibnamefont {Maze}},
  \bibinfo {author} {\bibfnamefont {C.}~\bibnamefont {Hermann-Avigliano}},\
  and\ \bibinfo {author} {\bibfnamefont {F.}~\bibnamefont {Barra}},\
  }\href@noop {} {\bibfield  {journal} {\bibinfo  {journal} {Physical Review
  E}\ }\textbf {\bibinfo {volume} {105}},\ \bibinfo {pages} {064119} (\bibinfo
  {year} {2022})}\BibitemShut {NoStop}%
\bibitem [{\citenamefont {Farina}\ \emph {et~al.}(2019)\citenamefont {Farina},
  \citenamefont {Andolina}, \citenamefont {Mari}, \citenamefont {Polini},\ and\
  \citenamefont {Giovannetti}}]{farina2019charger}%
  \BibitemOpen
  \bibfield  {author} {\bibinfo {author} {\bibfnamefont {D.}~\bibnamefont
  {Farina}}, \bibinfo {author} {\bibfnamefont {G.~M.}\ \bibnamefont
  {Andolina}}, \bibinfo {author} {\bibfnamefont {A.}~\bibnamefont {Mari}},
  \bibinfo {author} {\bibfnamefont {M.}~\bibnamefont {Polini}},\ and\ \bibinfo
  {author} {\bibfnamefont {V.}~\bibnamefont {Giovannetti}},\ }\href
  {https://doi.org/10.1103/PhysRevB.99.035421} {\bibfield  {journal} {\bibinfo
  {journal} {Physical Review B}\ }\textbf {\bibinfo {volume} {99}},\ \bibinfo
  {pages} {035421} (\bibinfo {year} {2019})}\BibitemShut {NoStop}%
\bibitem [{\citenamefont {Gherardini}\ \emph {et~al.}(2020)\citenamefont
  {Gherardini}, \citenamefont {Campaioli}, \citenamefont {Caruso},
  \citenamefont {Ruffo}, \citenamefont {Trombettoni}, \citenamefont {Buffoni},
  \citenamefont {M{"u}ller}, \citenamefont {Campisi},\ and\ \citenamefont
  {Smerzi}}]{gherardini2020stabilizing}%
  \BibitemOpen
  \bibfield  {author} {\bibinfo {author} {\bibfnamefont {S.}~\bibnamefont
  {Gherardini}}, \bibinfo {author} {\bibfnamefont {F.}~\bibnamefont
  {Campaioli}}, \bibinfo {author} {\bibfnamefont {F.}~\bibnamefont {Caruso}},
  \bibinfo {author} {\bibfnamefont {S.}~\bibnamefont {Ruffo}}, \bibinfo
  {author} {\bibfnamefont {A.}~\bibnamefont {Trombettoni}}, \bibinfo {author}
  {\bibfnamefont {L.}~\bibnamefont {Buffoni}}, \bibinfo {author} {\bibfnamefont
  {M.~M.}\ \bibnamefont {M{"u}ller}}, \bibinfo {author} {\bibfnamefont
  {M.}~\bibnamefont {Campisi}},\ and\ \bibinfo {author} {\bibfnamefont
  {A.}~\bibnamefont {Smerzi}},\ }\href
  {https://doi.org/10.1103/PhysRevResearch.2.013095} {\bibfield  {journal}
  {\bibinfo  {journal} {Physical Review Research}\ }\textbf {\bibinfo {volume}
  {2}},\ \bibinfo {pages} {013095} (\bibinfo {year} {2020})}\BibitemShut
  {NoStop}%
\bibitem [{\citenamefont {Qi}\ and\ \citenamefont {Jing}(2021)}]{qi2021magnon}%
  \BibitemOpen
  \bibfield  {author} {\bibinfo {author} {\bibfnamefont {S.-F.}\ \bibnamefont
  {Qi}}\ and\ \bibinfo {author} {\bibfnamefont {J.}~\bibnamefont {Jing}},\
  }\href {https://doi.org/10.1103/PhysRevA.104.032606} {\bibfield  {journal}
  {\bibinfo  {journal} {Physical Review A}\ }\textbf {\bibinfo {volume}
  {104}},\ \bibinfo {pages} {032606} (\bibinfo {year} {2021})}\BibitemShut
  {NoStop}%
\bibitem [{\citenamefont {Pokhrel}\ \emph {et~al.}(2025)\citenamefont
  {Pokhrel}, \citenamefont {Nokkala}, \citenamefont {Manzano}, \citenamefont
  {Hu}, \citenamefont {Sun},\ and\ \citenamefont
  {Plastina}}]{pokhrel2025dissipative}%
  \BibitemOpen
  \bibfield  {author} {\bibinfo {author} {\bibfnamefont {S.}~\bibnamefont
  {Pokhrel}}, \bibinfo {author} {\bibfnamefont {J.}~\bibnamefont {Nokkala}},
  \bibinfo {author} {\bibfnamefont {G.}~\bibnamefont {Manzano}}, \bibinfo
  {author} {\bibfnamefont {X.}~\bibnamefont {Hu}}, \bibinfo {author}
  {\bibfnamefont {J.}~\bibnamefont {Sun}},\ and\ \bibinfo {author}
  {\bibfnamefont {F.}~\bibnamefont {Plastina}},\ }\href
  {https://doi.org/10.1103/PhysRevLett.134.130401} {\bibfield  {journal}
  {\bibinfo  {journal} {Physical Review Letters}\ }\textbf {\bibinfo {volume}
  {134}},\ \bibinfo {pages} {130401} (\bibinfo {year} {2025})}\BibitemShut
  {NoStop}%
\bibitem [{\citenamefont {Seah}\ \emph {et~al.}(2021)\citenamefont {Seah},
  \citenamefont {Perarnau-Llobet}, \citenamefont {Haack}, \citenamefont
  {Brunner},\ and\ \citenamefont {Nimmrichter}}]{seah2021quantum}%
  \BibitemOpen
  \bibfield  {author} {\bibinfo {author} {\bibfnamefont {S.}~\bibnamefont
  {Seah}}, \bibinfo {author} {\bibfnamefont {M.}~\bibnamefont
  {Perarnau-Llobet}}, \bibinfo {author} {\bibfnamefont {G.}~\bibnamefont
  {Haack}}, \bibinfo {author} {\bibfnamefont {N.}~\bibnamefont {Brunner}},\
  and\ \bibinfo {author} {\bibfnamefont {S.}~\bibnamefont {Nimmrichter}},\
  }\href@noop {} {\bibfield  {journal} {\bibinfo  {journal} {Physical Review
  Letters}\ }\textbf {\bibinfo {volume} {127}},\ \bibinfo {pages} {100601}
  (\bibinfo {year} {2021})}\BibitemShut {NoStop}%
\bibitem [{\citenamefont {Zhu}\ \emph {et~al.}(2023)\citenamefont {Zhu},
  \citenamefont {Ebler}, \citenamefont {Dai}, \citenamefont {Liu},\ and\
  \citenamefont {Yadin}}]{zhu2023ico}%
  \BibitemOpen
  \bibfield  {author} {\bibinfo {author} {\bibfnamefont {G.}~\bibnamefont
  {Zhu}}, \bibinfo {author} {\bibfnamefont {D.}~\bibnamefont {Ebler}}, \bibinfo
  {author} {\bibfnamefont {J.-L.}\ \bibnamefont {Dai}}, \bibinfo {author}
  {\bibfnamefont {Z.}~\bibnamefont {Liu}},\ and\ \bibinfo {author}
  {\bibfnamefont {B.}~\bibnamefont {Yadin}},\ }\href
  {https://doi.org/10.1103/PhysRevLett.131.240401} {\bibfield  {journal}
  {\bibinfo  {journal} {Physical Review Letters}\ }\textbf {\bibinfo {volume}
  {131}},\ \bibinfo {pages} {240401} (\bibinfo {year} {2023})}\BibitemShut
  {NoStop}%
\bibitem [{\citenamefont {Rossini}\ \emph
  {et~al.}(2020{\natexlab{b}})\citenamefont {Rossini}, \citenamefont
  {Andolina}, \citenamefont {Rosa}, \citenamefont {Carrega},\ and\
  \citenamefont {Polini}}]{rossini2020sykadvantage}%
  \BibitemOpen
  \bibfield  {author} {\bibinfo {author} {\bibfnamefont {D.}~\bibnamefont
  {Rossini}}, \bibinfo {author} {\bibfnamefont {G.~M.}\ \bibnamefont
  {Andolina}}, \bibinfo {author} {\bibfnamefont {D.}~\bibnamefont {Rosa}},
  \bibinfo {author} {\bibfnamefont {M.}~\bibnamefont {Carrega}},\ and\ \bibinfo
  {author} {\bibfnamefont {M.}~\bibnamefont {Polini}},\ }\href
  {https://doi.org/10.1103/PhysRevLett.125.236402} {\bibfield  {journal}
  {\bibinfo  {journal} {Phys. Rev. Lett.}\ }\textbf {\bibinfo {volume} {125}},\
  \bibinfo {pages} {236402} (\bibinfo {year} {2020}{\natexlab{b}})},\ \Eprint
  {https://arxiv.org/abs/1912.07234} {arXiv:1912.07234 [quant-ph]} \BibitemShut
  {NoStop}%
\bibitem [{\citenamefont {Rosa}\ \emph {et~al.}(2020)\citenamefont {Rosa},
  \citenamefont {Rossini}, \citenamefont {Andolina}, \citenamefont {Polini},\
  and\ \citenamefont {Carrega}}]{rosa2020ultrastable}%
  \BibitemOpen
  \bibfield  {author} {\bibinfo {author} {\bibfnamefont {D.}~\bibnamefont
  {Rosa}}, \bibinfo {author} {\bibfnamefont {D.}~\bibnamefont {Rossini}},
  \bibinfo {author} {\bibfnamefont {G.~M.}\ \bibnamefont {Andolina}}, \bibinfo
  {author} {\bibfnamefont {M.}~\bibnamefont {Polini}},\ and\ \bibinfo {author}
  {\bibfnamefont {M.}~\bibnamefont {Carrega}},\ }\href
  {https://doi.org/10.1007/JHEP11(2020)067} {\bibfield  {journal} {\bibinfo
  {journal} {JHEP}\ }\textbf {\bibinfo {volume} {2020}}\bibfield  {number}
  {\bibinfo  {number} { (11)},\ \bibinfo {pages} {067}},\ }\Eprint
  {https://arxiv.org/abs/1912.07247} {arXiv:1912.07247 [quant-ph]} \BibitemShut
  {NoStop}%
\bibitem [{\citenamefont {Divi}\ \emph {et~al.}(2025)\citenamefont {Divi},
  \citenamefont {Murugan},\ and\ \citenamefont
  {Rosa}}]{divi2025randomwalkgraphs}%
  \BibitemOpen
  \bibfield  {author} {\bibinfo {author} {\bibfnamefont {F.}~\bibnamefont
  {Divi}}, \bibinfo {author} {\bibfnamefont {J.}~\bibnamefont {Murugan}},\ and\
  \bibinfo {author} {\bibfnamefont {D.}~\bibnamefont {Rosa}},\ }\href
  {https://doi.org/10.1103/PhysRevB.111.075138} {\bibfield  {journal} {\bibinfo
   {journal} {Phys. Rev. B}\ }\textbf {\bibinfo {volume} {111}},\ \bibinfo
  {pages} {075138} (\bibinfo {year} {2025})},\ \Eprint
  {https://arxiv.org/abs/2412.04560} {arXiv:2412.04560 [quant-ph]} \BibitemShut
  {NoStop}%
\bibitem [{\citenamefont {Konar}\ \emph {et~al.}(2022)\citenamefont {Konar},
  \citenamefont {Lakkaraju}, \citenamefont {Ghosh},\ and\ \citenamefont
  {Sen(De)}}]{konar2022ultracold}%
  \BibitemOpen
  \bibfield  {author} {\bibinfo {author} {\bibfnamefont {T.~K.}\ \bibnamefont
  {Konar}}, \bibinfo {author} {\bibfnamefont {L.~G.~C.}\ \bibnamefont
  {Lakkaraju}}, \bibinfo {author} {\bibfnamefont {S.}~\bibnamefont {Ghosh}},\
  and\ \bibinfo {author} {\bibfnamefont {A.}~\bibnamefont {Sen(De)}},\ }\href
  {https://doi.org/10.1103/PhysRevA.106.022618} {\bibfield  {journal} {\bibinfo
   {journal} {Phys. Rev. A}\ }\textbf {\bibinfo {volume} {106}},\ \bibinfo
  {pages} {022618} (\bibinfo {year} {2022})},\ \Eprint
  {https://arxiv.org/abs/2109.06816} {arXiv:2109.06816 [cond-mat.quant-gas]}
  \BibitemShut {NoStop}%
\bibitem [{\citenamefont {Grazi}\ \emph
  {et~al.}(2025{\natexlab{a}})\citenamefont {Grazi}, \citenamefont {Cavaliere},
  \citenamefont {Sassetti}, \citenamefont {Ferraro},\ and\ \citenamefont
  {Traverso~Ziani}}]{grazi2025chargingfreefermion}%
  \BibitemOpen
  \bibfield  {author} {\bibinfo {author} {\bibfnamefont {R.}~\bibnamefont
  {Grazi}}, \bibinfo {author} {\bibfnamefont {F.}~\bibnamefont {Cavaliere}},
  \bibinfo {author} {\bibfnamefont {M.}~\bibnamefont {Sassetti}}, \bibinfo
  {author} {\bibfnamefont {D.}~\bibnamefont {Ferraro}},\ and\ \bibinfo {author}
  {\bibfnamefont {N.}~\bibnamefont {Traverso~Ziani}},\ }\href
  {https://doi.org/10.1016/j.chaos.2025.116383} {\bibfield  {journal} {\bibinfo
   {journal} {Chaos, Solitons \& Fractals}\ }\textbf {\bibinfo {volume}
  {196}},\ \bibinfo {pages} {116383} (\bibinfo {year} {2025}{\natexlab{a}})},\
  \Eprint {https://arxiv.org/abs/2503.20424} {arXiv:2503.20424 [quant-ph]}
  \BibitemShut {NoStop}%
\bibitem [{\citenamefont {Grazi}\ \emph
  {et~al.}(2025{\natexlab{b}})\citenamefont {Grazi}, \citenamefont {Ferraro},\
  and\ \citenamefont {Traverso~Ziani}}]{grazi2025universal}%
  \BibitemOpen
  \bibfield  {author} {\bibinfo {author} {\bibfnamefont {R.}~\bibnamefont
  {Grazi}}, \bibinfo {author} {\bibfnamefont {D.}~\bibnamefont {Ferraro}},\
  and\ \bibinfo {author} {\bibfnamefont {N.}~\bibnamefont {Traverso~Ziani}},\
  }\href {https://doi.org/10.3390/en18236116} {\bibfield  {journal} {\bibinfo
  {journal} {Energies}\ }\textbf {\bibinfo {volume} {18}},\ \bibinfo {pages}
  {6116} (\bibinfo {year} {2025}{\natexlab{b}})},\ \Eprint
  {https://arxiv.org/abs/2511.16274} {arXiv:2511.16274 [quant-ph]} \BibitemShut
  {NoStop}%
\bibitem [{\citenamefont {Caravelli}\ \emph {et~al.}(2021)\citenamefont
  {Caravelli}, \citenamefont {Yan}, \citenamefont {Garcia-Pintos},\ and\
  \citenamefont {Hamma}}]{caravelli_energy_2021}%
  \BibitemOpen
  \bibfield  {author} {\bibinfo {author} {\bibfnamefont {F.}~\bibnamefont
  {Caravelli}}, \bibinfo {author} {\bibfnamefont {B.}~\bibnamefont {Yan}},
  \bibinfo {author} {\bibfnamefont {L.~P.}\ \bibnamefont {Garcia-Pintos}},\
  and\ \bibinfo {author} {\bibfnamefont {A.}~\bibnamefont {Hamma}},\ }\href
  {https://doi.org/10.22331/q-2021-07-15-505} {\bibfield  {journal} {\bibinfo
  {journal} {Quantum}\ }\textbf {\bibinfo {volume} {5}},\ \bibinfo {pages}
  {505} (\bibinfo {year} {2021})},\ \Eprint {https://arxiv.org/abs/2012.15026
  [quant-ph]} {2012.15026 [quant-ph]} \BibitemShut {NoStop}%
\bibitem [{\citenamefont {Zakavati}\ \emph {et~al.}(2021)\citenamefont
  {Zakavati}, \citenamefont {Tabesh},\ and\ \citenamefont
  {Salimi}}]{zakavati_bounds_2021}%
  \BibitemOpen
  \bibfield  {author} {\bibinfo {author} {\bibfnamefont {S.}~\bibnamefont
  {Zakavati}}, \bibinfo {author} {\bibfnamefont {F.~T.}\ \bibnamefont
  {Tabesh}},\ and\ \bibinfo {author} {\bibfnamefont {S.}~\bibnamefont
  {Salimi}},\ }\href {https://doi.org/10.1103/PhysRevE.104.054117} {\bibfield
  {journal} {\bibinfo  {journal} {Phys. Rev. E}\ }\textbf {\bibinfo {volume}
  {104}},\ \bibinfo {pages} {054117} (\bibinfo {year} {2021})}\BibitemShut
  {NoStop}%
\bibitem [{\citenamefont {Shaghaghi}\ \emph {et~al.}(2022)\citenamefont
  {Shaghaghi}, \citenamefont {Singh}, \citenamefont {Benenti},\ and\
  \citenamefont {Rosa}}]{shaghaghi2022micromasers}%
  \BibitemOpen
  \bibfield  {author} {\bibinfo {author} {\bibfnamefont {V.}~\bibnamefont
  {Shaghaghi}}, \bibinfo {author} {\bibfnamefont {V.}~\bibnamefont {Singh}},
  \bibinfo {author} {\bibfnamefont {G.}~\bibnamefont {Benenti}},\ and\ \bibinfo
  {author} {\bibfnamefont {D.}~\bibnamefont {Rosa}},\ }\href
  {https://doi.org/10.1088/2058-9565/ac8829} {\bibfield  {journal} {\bibinfo
  {journal} {Quantum Sci. Technol.}\ }\textbf {\bibinfo {volume} {7}},\
  \bibinfo {pages} {04LT01} (\bibinfo {year} {2022})},\ \Eprint
  {https://arxiv.org/abs/2204.09995} {arXiv:2204.09995 [quant-ph]} \BibitemShut
  {NoStop}%
\bibitem [{\citenamefont {Shaghaghi}\ \emph {et~al.}(2023)\citenamefont
  {Shaghaghi}, \citenamefont {Singh}, \citenamefont {Carrega}, \citenamefont
  {Rosa},\ and\ \citenamefont {Benenti}}]{shaghaghi2023lossy}%
  \BibitemOpen
  \bibfield  {author} {\bibinfo {author} {\bibfnamefont {V.}~\bibnamefont
  {Shaghaghi}}, \bibinfo {author} {\bibfnamefont {V.}~\bibnamefont {Singh}},
  \bibinfo {author} {\bibfnamefont {M.}~\bibnamefont {Carrega}}, \bibinfo
  {author} {\bibfnamefont {D.}~\bibnamefont {Rosa}},\ and\ \bibinfo {author}
  {\bibfnamefont {G.}~\bibnamefont {Benenti}},\ }\href
  {https://doi.org/10.3390/e25030430} {\bibfield  {journal} {\bibinfo
  {journal} {Entropy}\ }\textbf {\bibinfo {volume} {25}},\ \bibinfo {pages}
  {430} (\bibinfo {year} {2023})},\ \Eprint {https://arxiv.org/abs/2212.13417}
  {arXiv:2212.13417 [quant-ph]} \BibitemShut {NoStop}%
\bibitem [{\citenamefont {Rodr{\'i}guez}\ \emph {et~al.}(2023)\citenamefont
  {Rodr{\'i}guez}, \citenamefont {Rosa},\ and\ \citenamefont
  {Olle}}]{rodriguez2023aidiscovery}%
  \BibitemOpen
  \bibfield  {author} {\bibinfo {author} {\bibfnamefont {C.}~\bibnamefont
  {Rodr{\'i}guez}}, \bibinfo {author} {\bibfnamefont {D.}~\bibnamefont
  {Rosa}},\ and\ \bibinfo {author} {\bibfnamefont {J.}~\bibnamefont {Olle}},\
  }\href {https://doi.org/10.1103/PhysRevA.108.042618} {\bibfield  {journal}
  {\bibinfo  {journal} {Phys. Rev. A}\ }\textbf {\bibinfo {volume} {108}},\
  \bibinfo {pages} {042618} (\bibinfo {year} {2023})},\ \Eprint
  {https://arxiv.org/abs/2301.09408} {arXiv:2301.09408 [quant-ph]} \BibitemShut
  {NoStop}%
\bibitem [{\citenamefont {Rojo-Franc{\`a}s}\ \emph {et~al.}(2024)\citenamefont
  {Rojo-Franc{\`a}s}, \citenamefont {Isaule}, \citenamefont {Santos},
  \citenamefont {Juli{\'a}-D{\'i}az},\ and\ \citenamefont
  {Zinner}}]{rojofrancas2024stablecollective}%
  \BibitemOpen
  \bibfield  {author} {\bibinfo {author} {\bibfnamefont {A.}~\bibnamefont
  {Rojo-Franc{\`a}s}}, \bibinfo {author} {\bibfnamefont {F.}~\bibnamefont
  {Isaule}}, \bibinfo {author} {\bibfnamefont {A.~C.}\ \bibnamefont {Santos}},
  \bibinfo {author} {\bibfnamefont {B.}~\bibnamefont {Juli{\'a}-D{\'i}az}},\
  and\ \bibinfo {author} {\bibfnamefont {N.~T.}\ \bibnamefont {Zinner}},\
  }\href {https://doi.org/10.1103/PhysRevA.110.032205} {\bibfield  {journal}
  {\bibinfo  {journal} {Phys. Rev. A}\ }\textbf {\bibinfo {volume} {110}},\
  \bibinfo {pages} {032205} (\bibinfo {year} {2024})},\ \Eprint
  {https://arxiv.org/abs/2406.07397} {arXiv:2406.07397 [quant-ph]} \BibitemShut
  {NoStop}%
\bibitem [{\citenamefont {Mieghem}(2016)}]{van_mieghem_graph_2016}%
  \BibitemOpen
  \bibfield  {author} {\bibinfo {author} {\bibfnamefont {P.~V.}\ \bibnamefont
  {Mieghem}},\ }\href {https://arxiv.org/abs/1401.4580} {\bibinfo {title}
  {Graph eigenvectors, fundamental weights and centrality metrics for nodes in
  networks}} (\bibinfo {year} {2016}),\ \Eprint
  {https://arxiv.org/abs/1401.4580} {arXiv:1401.4580 [math.SP]} \BibitemShut
  {NoStop}%
\bibitem [{\citenamefont {Jiongsheng}\ and\ \citenamefont
  {Xinmao}(1998)}]{jiongsheng_lower_1998}%
  \BibitemOpen
  \bibfield  {author} {\bibinfo {author} {\bibfnamefont {L.}~\bibnamefont
  {Jiongsheng}}\ and\ \bibinfo {author} {\bibfnamefont {W.}~\bibnamefont
  {Xinmao}},\ }\href {https://doi.org/10.1007/BF02683829} {\bibfield  {journal}
  {\bibinfo  {journal} {Acta Mathematicae Applicatae Sinica}\ }\textbf
  {\bibinfo {volume} {14}},\ \bibinfo {pages} {443} (\bibinfo {year}
  {1998})}\BibitemShut {NoStop}%
\bibitem [{\citenamefont {de~Caen}(1998)}]{de_caen_upper_1998}%
  \BibitemOpen
  \bibfield  {author} {\bibinfo {author} {\bibfnamefont {D.}~\bibnamefont
  {de~Caen}},\ }\href {https://doi.org/10.1016/S0012-365X(97)00213-6}
  {\bibfield  {journal} {\bibinfo  {journal} {Discrete Mathematics}\ }\textbf
  {\bibinfo {volume} {185}},\ \bibinfo {pages} {245} (\bibinfo {year}
  {1998})}\BibitemShut {NoStop}%
\bibitem [{\citenamefont {Hagberg}\ \emph {et~al.}(2008)\citenamefont
  {Hagberg}, \citenamefont {Swart},\ and\ \citenamefont
  {S~Chult}}]{hagberg2008exploring}%
  \BibitemOpen
  \bibfield  {author} {\bibinfo {author} {\bibfnamefont {A.}~\bibnamefont
  {Hagberg}}, \bibinfo {author} {\bibfnamefont {P.}~\bibnamefont {Swart}},\
  and\ \bibinfo {author} {\bibfnamefont {D.}~\bibnamefont {S~Chult}},\
  }\href@noop {} {\emph {\bibinfo {title} {Exploring network structure,
  dynamics, and function using NetworkX}}},\ \bibinfo {type} {Tech. Rep.}\
  (\bibinfo  {institution} {Los Alamos National Lab.(LANL), Los Alamos, NM
  (United States)},\ \bibinfo {year} {2008})\BibitemShut {NoStop}%
\bibitem [{\citenamefont {Zhao}\ \emph {et~al.}(2025)\citenamefont {Zhao},
  \citenamefont {Zhao},\ and\ \citenamefont {Zhuang}}]{zhao2025non}%
  \BibitemOpen
  \bibfield  {author} {\bibinfo {author} {\bibfnamefont {S.-C.}\ \bibnamefont
  {Zhao}}, \bibinfo {author} {\bibfnamefont {Z.-R.}\ \bibnamefont {Zhao}},\
  and\ \bibinfo {author} {\bibfnamefont {N.-Y.}\ \bibnamefont {Zhuang}},\
  }\href@noop {} {\bibfield  {journal} {\bibinfo  {journal} {Physical Review
  E}\ }\textbf {\bibinfo {volume} {112}},\ \bibinfo {pages} {024129} (\bibinfo
  {year} {2025})}\BibitemShut {NoStop}%
\bibitem [{\citenamefont {Von~Collatz}\ and\ \citenamefont
  {Sinogowitz}(1957)}]{von_collatz_spektren_1957}%
  \BibitemOpen
  \bibfield  {author} {\bibinfo {author} {\bibfnamefont {L.}~\bibnamefont
  {Von~Collatz}}\ and\ \bibinfo {author} {\bibfnamefont {U.}~\bibnamefont
  {Sinogowitz}},\ }\href {https://doi.org/10.1007/BF02941924} {\bibfield
  {journal} {\bibinfo  {journal} {Abh.Math.Semin.Univ.Hambg.}\ }\textbf
  {\bibinfo {volume} {21}},\ \bibinfo {pages} {63} (\bibinfo {year}
  {1957})}\BibitemShut {NoStop}%
\bibitem [{\citenamefont {Oliveira}\ \emph {et~al.}(2013)\citenamefont
  {Oliveira}, \citenamefont {Oliveira}, \citenamefont {Justel},\ and\
  \citenamefont {Abreu}}]{oliveira_measures_2013}%
  \BibitemOpen
  \bibfield  {author} {\bibinfo {author} {\bibfnamefont {J.~A.~D.}\
  \bibnamefont {Oliveira}}, \bibinfo {author} {\bibfnamefont {C.~S.}\
  \bibnamefont {Oliveira}}, \bibinfo {author} {\bibfnamefont {C.}~\bibnamefont
  {Justel}},\ and\ \bibinfo {author} {\bibfnamefont {N.~M. M.~D.}\ \bibnamefont
  {Abreu}},\ }\href {https://doi.org/10.1590/S0101-74382013005000012}
  {\bibfield  {journal} {\bibinfo  {journal} {Pesqui. Oper.}\ }\textbf
  {\bibinfo {volume} {33}},\ \bibinfo {pages} {383} (\bibinfo {year}
  {2013})}\BibitemShut {NoStop}%
\bibitem [{\citenamefont {Nikiforov}(2006)}]{nikiforov_eigenvalues_2005}%
  \BibitemOpen
  \bibfield  {author} {\bibinfo {author} {\bibfnamefont {V.}~\bibnamefont
  {Nikiforov}},\ }\href {https://doi.org/10.48550/arXiv.math/0506257} {\bibinfo
  {title} {Eigenvalues and degree deviation in graphs}} (\bibinfo {year}
  {2006}),\ \Eprint {https://arxiv.org/abs/math/0506257} {math/0506257}
  \BibitemShut {NoStop}%
\bibitem [{\citenamefont {Aouchiche}\ \emph {et~al.}(2008)\citenamefont
  {Aouchiche}, \citenamefont {Bell}, \citenamefont {Cvetković}, \citenamefont
  {Hansen}, \citenamefont {Rowlinson}, \citenamefont {Simić},\ and\
  \citenamefont {Stevanović}}]{aouchiche_variable_2008}%
  \BibitemOpen
  \bibfield  {author} {\bibinfo {author} {\bibfnamefont {M.}~\bibnamefont
  {Aouchiche}}, \bibinfo {author} {\bibfnamefont {F.~K.}\ \bibnamefont {Bell}},
  \bibinfo {author} {\bibfnamefont {D.}~\bibnamefont {Cvetković}}, \bibinfo
  {author} {\bibfnamefont {P.}~\bibnamefont {Hansen}}, \bibinfo {author}
  {\bibfnamefont {P.}~\bibnamefont {Rowlinson}}, \bibinfo {author}
  {\bibfnamefont {S.~K.}\ \bibnamefont {Simić}},\ and\ \bibinfo {author}
  {\bibfnamefont {D.}~\bibnamefont {Stevanović}},\ }\href
  {https://doi.org/10.1016/j.ejor.2006.12.059} {\bibfield  {journal} {\bibinfo
  {journal} {European Journal of Operational Research}\ }\textbf {\bibinfo
  {volume} {191}},\ \bibinfo {pages} {661} (\bibinfo {year}
  {2008})}\BibitemShut {NoStop}%
\end{thebibliography}%

\appendix
\onecolumngrid
\section{Linear-fractional form of the early-time power ratio}
\label{sec:linfrac}

    \begin{proposition}[Linear-fractional form and curvature]
    \label{prop:linfrac}
        Fix an open interval $I\subset\mathbb{R}$ on which the active negative set $K_-(h)$ is constant and nonempty. Then
        \begin{equation}
        \label{eq:linfrac-form}
            R_G(h)=\frac{A_0+h\,S_0}{C_0+h\,K_0},
        \end{equation}
        with
        \begin{equation}
            A_0:=\sum_{k\in K_-} w_k\,\frac{N}{M}\,\epsilon_k,\quad
            S_0:=\sum_{k\in K_-} w_k,\quad
            C_0:=\sum_{k\in K_-} \frac{N}{M}\,\epsilon_k,\quad
            K_0:=|K_-|.
        \end{equation}
        Let $D(h):=C_0+hK_0=\sum_{k\in K_-}E_k(h)$. On $I$ one has $D(h)<0$, and
        \begin{equation}
        \label{eq:linfrac-derivs}
            \frac{dR_G}{dh}(h)=\frac{S_0C_0-A_0K_0}{D(h)^2},
            \qquad
            \frac{d^2R_G}{dh^2}(h)=-\frac{2K_0\,(S_0C_0-A_0K_0)}{D(h)^3}.
        \end{equation}
        In particular, $\tfrac{dR_G}{dh}$ has a constant sign on $I$, so $R_G$ is monotone there. Changes in slope (and thus curvature) can occur only at knots where some $E_k(h)=0$ and the set $K_-(h)$ changes.
        \end{proposition}

        \begin{proof}
        On $I$, the active index set $K_-(h)$ is fixed and $E_k(h)=(N/M)\epsilon_k+h$ for $k\in K_-$. Hence
        \begin{equation}
            R_G(h)=\frac{\sum_{k\in K_-} w_k\,E_k(h)}{\sum_{k\in K_-}E_k(h)}
            =\frac{\sum_{k\in K_-} w_k\big(\tfrac{N}{M}\epsilon_k+h\big)}{\sum_{k\in K_-}\big(\tfrac{N}{M}\epsilon_k+h\big)}
            =\frac{A_0+hS_0}{C_0+hK_0},
        \end{equation}
        giving \eqref{eq:linfrac-form}. Since $k\in K_-$ means $E_k(h)<0$, summing yields $D(h)=\sum_{k\in K_-}E_k(h)<0$ on $I$.
        Differentiate the linear--fractional form:
        \begin{equation}
            \frac{dR_G}{dh}(h)=\frac{S_0(C_0+hK_0)-(A_0+hS_0)K_0}{(C_0+hK_0)^2}
            =\frac{S_0C_0-A_0K_0}{\big(C_0+hK_0\big)^2}.
        \end{equation}
        Differentiating again and using $D'(h)=K_0$ gives
        \begin{equation}
            \frac{d^2R_G}{dh^2}(h)=-\,\frac{2K_0\,(S_0C_0-A_0K_0)}{\big(C_0+hK_0\big)^3},
        \end{equation}
        which proves \eqref{eq:linfrac-derivs}. Since the denominator in the first derivative is $>0$ on $I$, the sign is constant there, so $R_G$ is monotone. Knots occur only when some $E_k(h)$ crosses $0$ and $K_-(h)$ changes.
        \end{proof}

        \begin{remark}
        Because $D(h)<0$ on $I$, the first and second derivatives in \eqref{eq:linfrac-derivs} have the \emph{same} sign: $R_G$ is increasing and convex if $S_0C_0-A_0K_0>0$, decreasing and concave if $S_0C_0-A_0K_0<0$, and constant if $S_0C_0-A_0K_0=0$. It is convenient to rewrite
        \begin{equation}
            S_0C_0-A_0K_0
            =\frac{N}{M}\,S_0K_0\Big(\bar{\epsilon}-\bar{\epsilon}_w\Big),
            \qquad
            \bar{\epsilon}:=\frac{1}{K_0}\sum_{k\in K_-}\epsilon_k,\quad
            \bar{\epsilon}_w:=\frac{1}{S_0}\sum_{k\in K_-} w_k\,\epsilon_k.
        \end{equation}
        Thus the slope and curvature are positive exactly when the weights $w_k$ emphasize the \emph{more negative} eigenvalues ($\bar{\epsilon}_w<\bar{\epsilon}$), negative when they emphasize the \emph{less negative} ones, and zero when the weighted and unweighted means coincide. For the star $S_N$, $|K_-|=1$ on a wide $h$-window; then $\bar{\epsilon}=\bar{\epsilon}_w$ and $R_{S_N}(h)$ is constant on that window. This means that inclusion of the local energy term in the star graph Hamiltonian is inconsequential from the power point of view, at least in the linear order. One could naively hope that $\frac{\text{d}R_G}{\text{d}h}$ is non-positive for other architectures, so that the inclusion of the perturbatively small diagonal term never increases the charging power. This is however \emph{not} the case. Our envelope argument will allow us to `prove' optimality of the star graph for all $h$ independently of this piecewise behavior in $h$.
        \end{remark}

\section{Random graph ensembles}
\label{sec:random-graph-ensembles}

    We briefly recall the precise generative rules used in our experiments; all graphs are simple, undirected, and unweighted (no self-loops, no multi-edges). The adjacency matrix $A$ is symmetric with zero diagonal.
    \begin{description}
        \item[ER $G(N,p)$.] On the vertex set $[N]=\{1,\dots,N\}$, each unordered pair $\{i,j\}$ with $i<j$ is included as an edge independently with probability $p\in(0,1)$. We take $p\in\{0.10,0.20\}$. This yields the classical Erd\H{o}s--R\'enyi model.
        \item[Uniform random tree.] A labeled tree on $[N]$ is sampled uniformly from the $N^{N-2}$ possibilities (Cayley). Operationally, we draw a Pr\"ufer sequence of length $N-2$ with i.i.d.\ uniform entries in $[N]$, then reconstruct the tree from the sequence; the resulting tree is uniform over all labeled trees.
        \item[Barab\'asi--Albert (BA) with parameter $m$.] Start from the complete graph $K_{m+1}$. For $t=m+2,\dots,N$, add a new vertex and connect it to $m$ distinct existing vertices, chosen without replacement with probability proportional to their current degrees (linear preferential attachment). Self-loops and parallel edges are disallowed; if a sampled endpoint repeats, we resample until $m$ distinct neighbors are obtained.
        \item[Unbalanced stochastic block model (SBM).] Partition $[N]=U\cup V$ with $|U|=\lfloor N/4\rfloor$ and $|V|=N-|U|$. For each unordered pair within $U$ or within $V$, place an edge independently with probability $p_{\mathrm{in}}=0.05$; for each cross pair $\{u,v\}$ with $u\in U$, $v\in V$, place an edge with probability $p_{\mathrm{out}}=0.9$. All edge decisions are independent. This yields an approximately bipartite (but noisy) structure with strong cross-part connectivity.
    \end{description}

\section{Charging dependence on microscopic parameters}\label{sec: micro params}
The accumulated work at time $t$ is given by
\begin{equation}
    W(t) = \text{Tr}[\rho_\textsc{b}(t) H_\textsc{b}] - \text{Tr}[\rho_\textsc{b}(0) H_\textsc{b}],
\end{equation}
while the average charging power is naturally defined as
\begin{equation}
    P(t) = \frac{W(t)}{t}.
\end{equation}

These quantities provide a direct picture of the energy transfer process from the cavity to the battery. Figure~\ref{fig:charge_dynamics} and \ref{fig:charge_dynamics1} show representative examples of $W(t)$ and $P(t)$ for a star interaction topology and different values of the cavity frequency $\omega$ and $\kappa$.

\begin{figure}[h!]
    \centering
    \includegraphics[width=0.47\textwidth]{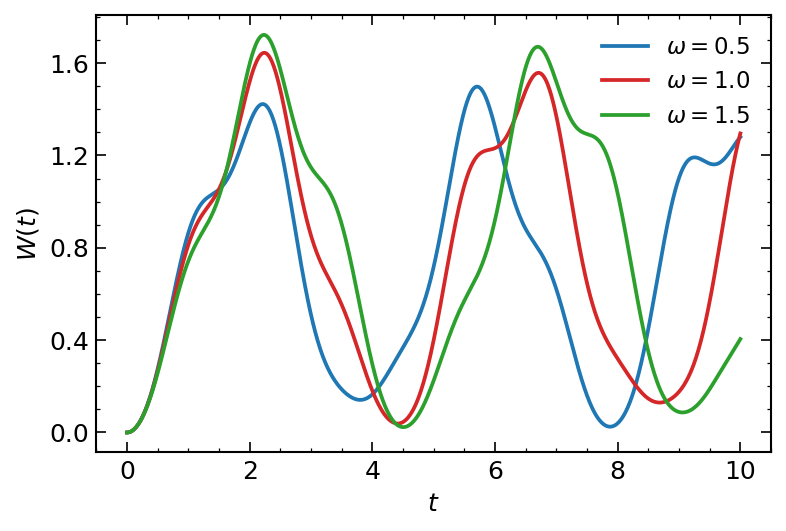}
    \includegraphics[width=0.47\textwidth]{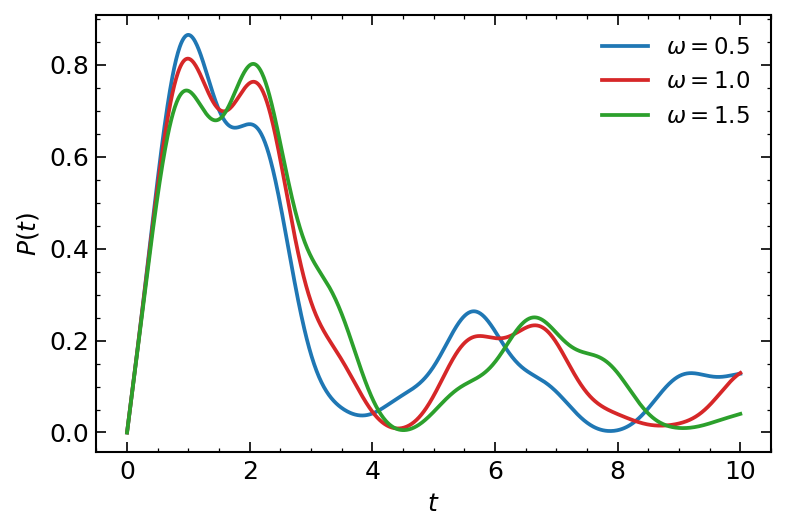}
    \hfill
    \caption{\textbf{Time evolution of the accumulated work (left) and the average charging power (right)} for a star topology with $N = 7$, \(\kappa=1\). The different curves correspond to various cavity frequencies $\omega$. }
    \label{fig:charge_dynamics}

\end{figure}

\begin{figure}[h!]
    \centering
    \includegraphics[width=0.47\textwidth]{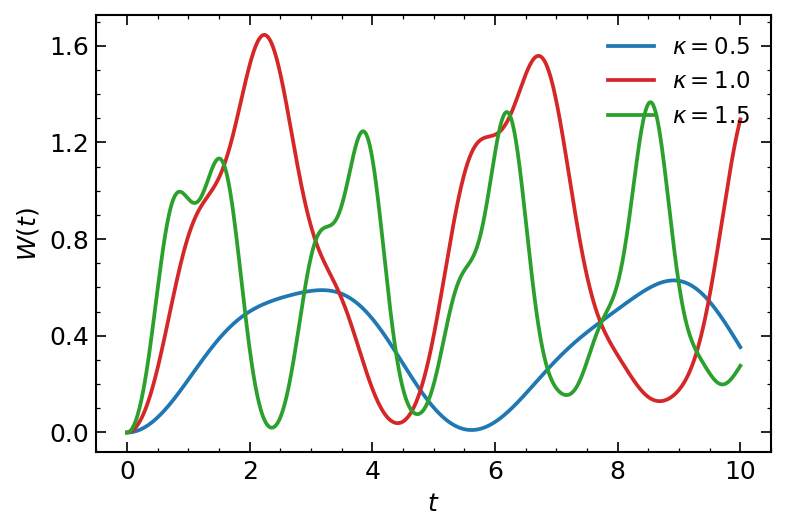}
    \hfill
    \includegraphics[width=0.48\textwidth]{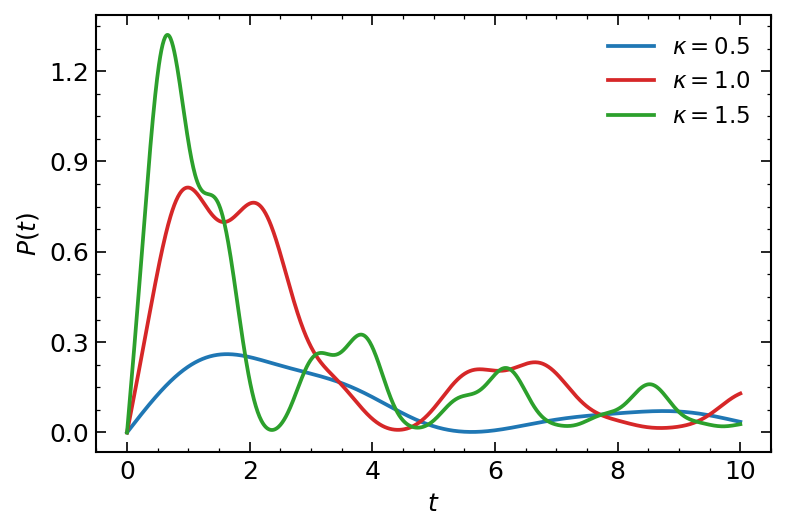}
    \caption{\textbf{Time evolution of the accumulated work (left) and the average charging power (right)} for a star topology with $N = 7$, with \(\omega=1\). The different curves correspond to various coupling values $\kappa$.}
    \label{fig:charge_dynamics1}
\end{figure}

The charging dynamics clearly highlights distinct roles played by the frequency \(\omega\) and the coupling strength \(\kappa\).
Varying \(\omega\) mainly affects the amplitude and structure of the oscillations in the accumulated work \(W(t)\), while leaving the overall charging timescale essentially unchanged.
In particular, changing \(\omega\) shifts the value of the maximal charging power, but has little influence on the time at which this maximum is reached.
This indicates that the charger frequency primarily controls the magnitude of energy exchange and the oscillatory features of the dynamics, rather than the speed at which the battery is charged.

In contrast, varying the coupling strength \(\kappa\) leads to much more pronounced modifications of the charging process.
Different values of \(\kappa\) can drastically change both the oscillatory pattern and the amplitude of \(W(t)\), resulting in qualitatively distinct charging behaviors.
This strong sensitivity is even more apparent in the power \(P(t)\), whose profile is significantly reshaped as \(\kappa\) varies.
In particular, modifying \(\kappa\) not only alters the maximal charging power, but also shifts the time at which this maximum is attained, and can strongly affect the efficiency with which energy is retained in the battery.
Overall, these results show that while \(\omega\) mainly tunes oscillation amplitudes and peak values, the coupling strength \(\kappa\) acts as the dominant parameter controlling both the intensity and the timescale of the charging dynamics.

Let us mention at this point that the oscillations in the stored energy could of course be mitigated by providing a more refined charging protocol, for instance by including dissipation in order to stabilize the asymptotic dynamics \cite{zhao2025non}.

\section{Power for exemplary graphs}
To exemplify the influence of interaction structure, we compute the average charging power $P(t)$ for four distinct graph topologies at fixed system size $N=7$.

The comparison is shown in Fig.~\ref{fig:standard_topologies}. The star topology (green) clearly outperforms all other configurations, achieving both the highest peak power and the largest overall charging performance. In contrast, the complete graph (blue) yields a negligible charging power when the battery is initialized in its ground state. The linear chain (red) and the random Erdős--Rényi graph (purple) display intermediate performance, but remain well below that of the star architecture.

This comparison highlights the crucial role of interaction structure: strongly centralized, star-like topologies enable much more efficient energy transfer than either homogeneous (complete) or weakly connected (chain, random) graphs.

\begin{figure}[h]
    \centering
    \includegraphics[width=0.47\linewidth]{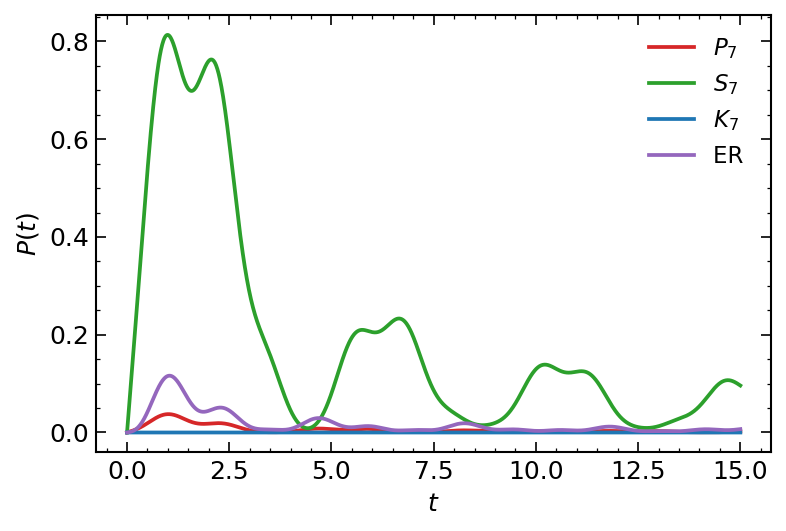}
    \caption{\textbf{Comparison of charging performance across graph families.} Average power $P(t)$ for $N=7$, with \(h=0\), \(\omega=\kappa=1\) . The star (green) achieves the best performance, while the complete graph (blue) fails to charge. The linear chain (red) and random Erdős–Rényi graph ($p=0.4$, purple) yield intermediate results.}
    \label{fig:standard_topologies}
\end{figure}

\section{Power versus independence number}
In our quest of finding the characteristic property of a graph that determines the charging power of the associated battery model, we investigated the graph’s so-called independence number $\alpha(G)$.

The independence number $\alpha(G)$ is defined as the size of the largest subset of vertices such that no two vertices in the subset are connected by an edge. It quantifies the maximum number of sites that can be mutually non-interacting within the graph. As such, $\alpha(G)$ provides a meaningful measure of the level of constraint imposed by the interaction topology.

Figure~\ref{fig:Pmax_alpha} shows $P_{\max}$ as a function of $\alpha(G)$ for all graphs at fixed system size $N=7$, for both $h=0$ and $h=1$. Although, when $h=1$, the values of $P_{\max}$ are reduced for all graphs, the behaviour remains identical for both cases.

\begin{figure}[h]
    \centering
    \begin{subfigure}{0.48\textwidth}
        \centering
        \includegraphics[width=\textwidth]{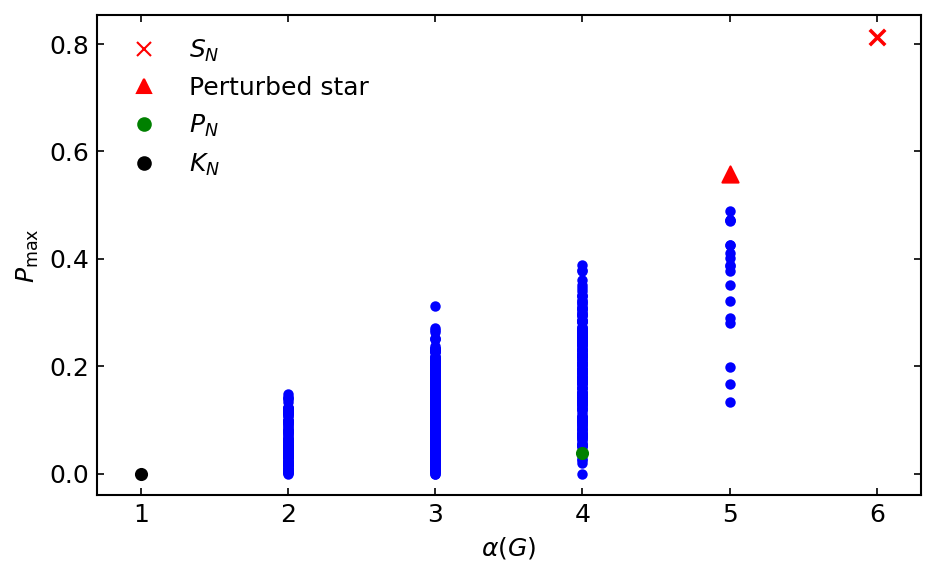}
    \end{subfigure}
    \hfill
    \begin{subfigure}{0.48\textwidth}
        \centering
        \includegraphics[width=\textwidth]{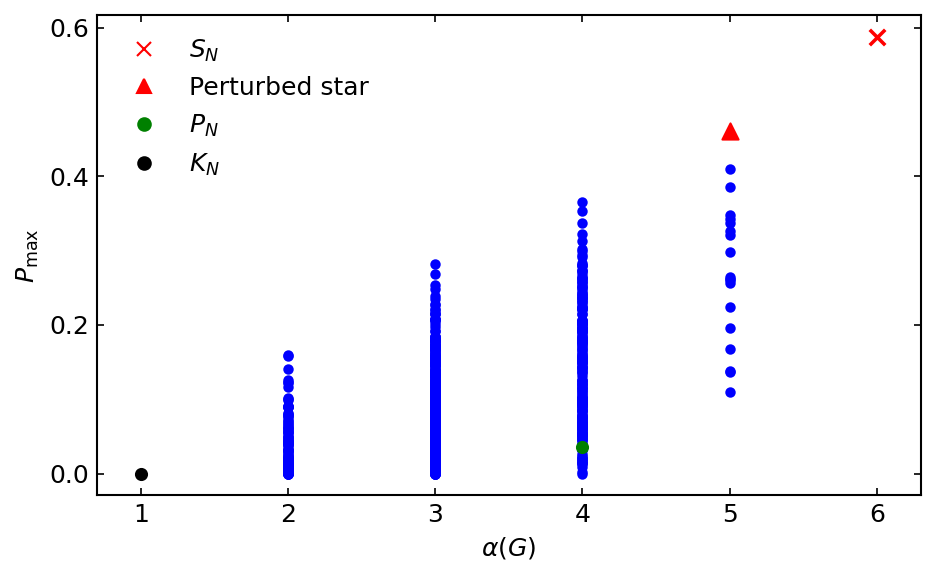}
    \end{subfigure}
    \caption{\textbf{$P_{\max}$ as a function of the independence number $\alpha(G)$ for all graphs at fixed system size}. Left: $h=0$. Right: $h=1$. Canonical graph topologies are highlighted: star ($S_N$), perturbed star, path ($P_N$), and complete graph ($K_N$). With \(\omega=\kappa=1\).}
    \label{fig:Pmax_alpha}
\end{figure}

A clear correlation between $P_{\max}$ and the independence number $\alpha(G)$ is observed: graphs with larger independence numbers systematically achieve higher charging power. In particular, the star topology $S_N$, which maximizes $\alpha(G)$, consistently yields the largest values of $P_{\max}$, followed by the perturbed star topology.

These results indicate that the independence number captures a key structural property governing the charging efficiency of the quantum battery. In particular, a larger $\alpha(G)$ allows for a greater number of weakly constrained sites, facilitating energy absorption and redistribution, whereas highly connected graphs impose strong constraints that hinder the charging process.

\section{Extra analytical results}

            We prove here a lemma which shows that $w_k$ cannot be arbitrarily large for modes with energies closer to zero. In fact it may grow at most as fast with energy as its associated weight $p_k(h)$ decreases. 
        It is hence consistent with $R_G(h)$ being dependent mostly on $w_\textrm{min}$.

            \begin{lemma}[Bound on uniform overlap growth]\label{lemma: w_k}
            Let $\epsilon_\textrm{max} = \max_{k} \epsilon_k$. For any $k$ such that $\epsilon_k\neq \epsilon_\textrm{max}$ we have
            \begin{equation}
                w_k \leq N \frac{\epsilon_\textrm{max} - \frac{2M}{N}}{\epsilon_\textrm{max} - \epsilon_k }.
            \end{equation}
        \end{lemma}
        \begin{proof}
            Consider $x =\one - v_k u_k$, and the Rayleigh quotient
            \begin{equation}
                \frac{x^\top A x}{\norm{x}^2} = \frac{2M-\epsilon_k w_k}{N-w_k},
            \end{equation}
            where we used that $\one^\top A \one = 2M$, that $u_k$ is a normalised eigenvector corresponding to $\epsilon_k$, and the definitions of $v_k$ and $w_k$. The claim now follows from the fact that the Rayleigh quotient is upper bounded by $\epsilon_\mathrm{max}$.
        \end{proof}

        We note that Lemma~\ref{lemma: w_k} does not provide a strong enough bound on $w_\textrm{min}$ for proving the conjecture~\ref{conjecture}, with pineapple graph $P(N,\lfloor\tfrac{N}{2}\rfloor +1)$ being a counterexample. In the case of the least eigenvalue, the bound from Lemma~\ref{lemma: w_k} is proportional to the quotient of the Collatz-Sinogowitz (CS) graph irregularity measure \cite{von_collatz_spektren_1957, oliveira_measures_2013, nikiforov_eigenvalues_2005} and the adjacency matrix spectral spread. The CS irregularity measure is conjectured to be maximised among connected graphs with $N>10$ exactly by the aforementioned pineapple graphs \cite{aouchiche_variable_2008, oliveira_measures_2013}.

\end{document}